\documentclass[a4paper,UKenglish,cleveref, autoref, thm-restate]{lipics-v2021}
\usepackage{amsthm,amsmath}
\usepackage{mathtools}
\usepackage{xspace}
\usepackage{enumerate}
\usepackage{tcolorbox}
\usepackage{mathrsfs}

\DeclareMathOperator{\Aut}{Aut}
\DeclareMathOperator{\End}{End}
\DeclareMathOperator{\Pol}{Pol}
\DeclareMathOperator{\Inv}{Inv}
\DeclareMathOperator{\Orb}{Orb}

\newcommand{\cclone}[1]{\ensuremath{\langle #1 \rangle}}

\newcommand{\Nat}{\mathbb{N}}
\newcommand{\Int}{\mathbb{Z}}
\newcommand{\Rat}{\mathbb{Q}}

\def \cA {\mathcal{A}}
\def \cB {\mathcal{B}}
\def \cC {\mathcal{C}}

\newcommand{\ar}{\ensuremath{\mathrm{ar}}}

\nolinenumbers
\newcommand{\csp}[1]{\textsc{CSP}\ensuremath{(#1)}}

\newcommand{\problemDef}[3]
{%
    
    \begin{tcolorbox}[arc=0.1mm,boxsep=-0.6mm,left=1.9mm,right=1.9mm,bottom=1.4mm,top=1.4mm,adjusted title={\strut \sc#1},colback=white!5]

    \noindent\textbf{Instance:} #2
    
    \noindent\textbf{Question:} #3
    \end{tcolorbox}
}

\newcommand{\notes}[1]{}

\newcommand{\first}[1]{$1^{\mathrm{st}}$}
\newcommand{\second}[1]{$2^{\mathrm{nd}}$}

\newcommand{\squishlisttwo}{
\begin{list}{$\blacktriangleright$}
{ \setlength{\itemsep}{0.5pt}
\setlength{\parsep}{0pt}
\setlength{\topsep}{0pt}
\setlength{\partopsep}{0.5pt}
\setlength{\leftmargin}{1em}
\setlength{\labelwidth}{1em}
\setlength{\labelsep}{0.5em} } }

\newcommand{\squishend}{
\end{list} }

\newcommand{\cc}[1]{{\mbox{\textnormal{\textsf{#1}}}}\xspace}  %

\newcommand{\Poly}{\cc{P}}
\newcommand{\NP}{\cc{NP}}
\newcommand{\coNP}{\cc{coNP}}

\newcommand{\FPT}{\cc{FPT}}

\newcommand{\paraNP}{\cc{pNP}}

\newcommand{\alien}[3]{{\ensuremath{\textsc{CSP}_{\leq #1}(#2 \cup #3)}}}

\newcommand{\balien}[2]{{\ensuremath{\textsc{CSP}_{\leq #1}(#2)}}}
\newcommand{\Sol}{\textrm{Sol}}
\newcommand{\NEQQ}[1]{\ensuremath{\mathrm{NEQ}_{#1}}}

\newcommand{\numac}{\ensuremath{\#{\sf ac}}}

\usepackage{todonotes}
\usepackage{comment}

\definecolor{OliveGreen}{rgb}{0.55, 0.64, 0.40}

\title{CSPs with Few Alien Constraints}
\author{Peter Jonsson}{Department of Computer and Information Science, Link{\"o}ping University, Link{\"o}ping, Sweden}{peter.jonsson@liu.se}{}{}
\author{Victor Lagerkvist}{Department of Computer and Information Science, Link{\"o}ping University, Link{\"o}ping, Sweden}{victor.lagerkvist@liu.se}{}{}
\author{George Osipov}{Department of Computer and Information Science, Link{\"o}ping University, Link{\"o}ping, Sweden}{george.osipov@liu.se}{}{}

\authorrunning{P. Jonsson, V. Lagerkvist, and G. Osipov} %

\Copyright{Peter Jonsson, Victor Lagerkvist, and George Osipov} %



\category{} %

\funding{The first author is partially supported by the Swedish research council under grant VR-2021-04371. The second author is partially supported by the Swedish research council under grant VR-2022-03214.
The first and third author was supported by
the Wallenberg AI, Autonomous Systems and Software Program (WASP) funded
by the Knut and Alice Wallenberg Foundation.}

\relatedversion{This is an extended version of {\em CSPs with Few Alien Constraints}, appearing in the proceedings of the {\em 30th Conference on Principles and Practice of Constraint Programming} (CP-2024).}

\begin{document}

\maketitle

\begin{abstract}
The \emph{constraint satisfaction problem} asks to decide if a set of constraints 
over a relational structure $\mathcal{A}$ is satisfiable (CSP$(\mathcal{A})$). We 
consider CSP$(\mathcal{A} \cup \mathcal{B})$ where $\mathcal{A}$ is a structure and $\mathcal{B}$ 
is an \emph{alien} structure, and analyse its (parameterized)
complexity when at most $k$ alien constraints are allowed. We establish connections and obtain transferable complexity results to several well-studied problems that previously escaped classification attempts.
Our novel approach, utilizing logical and algebraic methods, yields an \textsf{FPT} versus \textsf{pNP} dichotomy for arbitrary finite structures and sharper dichotomies for Boolean structures and first-order reducts of $(\mathbb{N},=)$ (equality CSPs), together with many partial results for general $\omega$-categorical structures.
\end{abstract}

\newpage

\section{Introduction}

The \emph{constraint satisfaction problem} over a structure $\cA$ ($\csp{\cA}$) is the problem of verifying whether a set of constraints over $\cA$ admits at least one solution. This problem framework is vast, and, just to name a few, include all Boolean satisfiability problems as well as $k$-coloring problems, and for infinite domains we may formulate both problems centrally related to model checking first-order formulas and qualitative reasoning.
A notable example where a complete complexity dichotomy is known (separating tractable from NP-hard problems) includes \emph{all} finite structures~\cite{Bulatov:focs2017,Zhuk:jacm2020} which settled the {\em Feder-Vardi conjecture}~\cite{DBLP:journals/siamcomp/FederV98} in the positive. For infinite domains we may e.g. mention
first-order definable relations over well-behaved base structures like $(\mathbb{N}, =)$ and $(\mathbb{Q}, <)$~\cite{Bodirsky:Book}. While impressive mathematical achievements, these dichotomy results are still somewhat unsatisfactory from a practical perspective since we are unlikely to encounter instances which are based on \emph{purely} tractable constraints. Could it be possible to extend the reach of these powerful theoretical results by relaxing the basic setting so that we may allow greater flexibility than purely tractable constraints while still obtaining something simpler than an arbitrary NP-hard CSP?

We consider this problem in a \emph{hybrid} setting via problems of the form $\csp{\cA \cup \cB}$ where $\cA$ is a ``stable'', tractable background structure and $\cB$ is an \emph{alien} structure. We focus on the case when $\csp{\cA \cup \cB}$ is NP-hard (thus, richer than a polynomial-time solvable problem) but where we have comparably few constraints from the alien structure $\mathcal{B}$. This problem is compatible with the influential framework of \emph{parameterized complexity} which has been used with great effect to study \emph{structurally} restricted problems (e.g., based on tree-width) but where comparably little is known when one simultaneously restricts the allowed constraints. 

We begin (in Section~\ref{sec:alien-applications}) by relating the CSP problem with alien constraints to other problems, namely, (1) \emph{model checking}, (2) the problem of checking whether a constraint in a CSP instance is \emph{redundant}, (3) the \emph{implication} problem and (4) the \emph{equivalence} problem. We prove that the latter three problems are equivalent under Turing reductions and provide a general method for obtaining complexity dichotomies for all of these problems via a complexity dichotomy for the CSP problem with alien constraints. Importantly, all of these problems are well-known in their own right, but have traditionally been studied with wildly disparate tools and techniques, but by viewing them under the unifying lens of alien constraints we not only get four dichotomies for the price of one but also open the powerful toolbox based on \emph{universal algebra}. For non-Boolean domains this is not only a simplifying aspect but an absolute necessity to obtain general results. We expand upon the algebraic approach in Section~\ref{sec:generalresults} and relate alien constraints to \emph{primitive positive definitions} (pp-definitions) and the important  notion of a \emph{core}. %
As a second general contribution we explore the case when each relation in $\cB$ can be defined via an \emph{existential positive formula} over $\cA$. This results in a general \emph{fixed-parameter tractable} (FPT) algorithm (with respect to the number of alien constraints) applicable to both finite, and, as we demonstrate later, many natural classes of structures over infinite domains.

In the second half of the paper we attack the complexity of alien constraints more systematically. We begin with structures over finite domains where we obtain a general tractability result by combining the aforementioned FPT algorithm together with the CSP dichotomy theorem \cite{Bulatov:focs2017,Zhuk:jacm2020}. In a similar vein we obtain a general hardness result based on a universal algebraic gadget. Put together this yields a  general result: if $\cA \cup \cB$ is a core (which we may assume without loss of generality) then either $\alien{}{\cA}{\cB}$ is FPT, or $\alien{p}{\cA}{\cB}$ is NP-hard for some $p \geq 0$, i.e., is \emph{para-NP-hard} ($\paraNP$-hard). Thus, from a parameterized complexity view we obtain a complete dichotomy (FPT versus $\paraNP$-hardness) for finite-domain structures. However, to also obtain dichotomies for implication, equivalence, and the redundancy problem, we need sharper bounds on the parameter $p$.
We concentrate on two special cases. We begin with Boolean structures in Section~\ref{sec:class-boolean} and  obtain a complete classification which e.g.\ states that  $\alien{}{\cA}{\cB}$ is FPT if $\cA$ is in one of the classical \emph{Schaefer} classes, and give a precise characterization of $\alien{p}{\cA}{\cB}$ for all relevant values of $p$ if $\cA$ is not Schaefer. For example, if we assume that $\cA$ is Horn, we may thus conclude that $\alien{}{\cA}{\cB}$ is \FPT for \emph{any} alien Boolean structure $\cB$. More generally this dichotomy is sufficiently sharp to also yield dichotomies for implication, equivalence, and redundancy. Compared to the proofs by Schnoor \& Schnoor~\cite{schnoor2008a} for implication and B\"ohler~\cite{Boehler:etal:csl2002} for equivalence, we do not use an exhaustive case analysis over Post's lattice.

In Section~\ref{sec:infindom-languages} we consider structures over infinite domains. 
If we assume that $\cA$ and $\cB$ are $\omega$-categorical, then we manage to lift the FPT algorithm based on existential positive definability from Section~\ref{sec:generalresults} to the infinite setting. Another important distinction is that the notion of a core, and subsequently the common trick of singleton expansion, works differently for $\omega$-categorical languages. Here we follow Bodirsky~\cite{Bodirsky:Book} and use the notion of a \emph{model-complete core}, which means that all $n$-ary orbits are pp-definable, where an orbit is defined as the action of the automorphism group over a fixed $n$-ary tuple. This allows us to, for example, prove that $\alien{}{\cA}{\cB}$ is FPT whenever $\cA$ is an $\omega$-categorical model-complete core and $\csp{\cA}$ is in P such that the orbits of the automorphism group of $\cB$ are included in the orbits of the automorphism group of $\cA$. This forms a cornerstone for the dichotomy for equality languages since the only remaining cases are when $\cA$ is $0$-valid (meaning that each relation contains a constant tuple) but not Horn (defined similarly to the Boolean domain), and when $\cB$ is not 0-valid. The remaining cases are far from trivial, however, and we require the algebraic machinery from Bodirsky et al.~\cite{bodirsky_chen_pinsker_2010} which provides a characterization of equality languages in terms of their \emph{retraction} to finite domains. We rely on this description via a recent classification result by Osipov \& Wahlstr\"om~\cite{equalityMinCSParXiv}. Importantly, our dichotomy result is sufficiently sharp to additionally obtain complexity dichotomies for the implication, equivalence, and redundancy problems. To the best of our knowledge, these dichotomies are the first of their kind for arbitrary equality languages.

We finish the paper with a comprehensive discussion in Section~\ref{sec:discussion}. Most importantly, we have opened up the possibility to systematically study not only alien constraints, but also related problems that have previously escaped complexity classifications. For future research the main open questions are whether (1) sharper results can be obtained for arbitrary finite domains and (2) which further classes of infinite domain structures should be considered. 

Proofs of statements marked with $(\star)$ can be found in the appendix in the end of the paper.

\section{Preliminaries} \label{sec:prelims}

We begin by introducing the basic terminology and the fundamental problems under consideration. We assume throughout the paper that the complexity classes
\Poly and \NP are distinct. We let $\Rat$ denote the rationals, $\Nat=\{0,1,2,\dots\}$ the natural numbers, 
$\Int=\{\dots, -2.-1,0,1,2, \dots\}$ the integers, and $\Int_+=\{1,2,3,\dots\}$ the positive integers.
For every $c \in \Int_+$, we let $[c]=\{1,2,\dots,c\}$.

A \emph{parameterized problem} is a subset of $\Sigma^* \times \Nat$ where $\Sigma$ is the input alphabet, i.e., an instance is given by $x \in \Sigma^*$ of size $n$ and a natural number $k$, and the running time of an algorithm is studied with respect to both $k$ and $n$. The most favourable complexity class is \FPT\
(\emph{fixed-parameter tractable}),
which contains all problems that can be decided 
in $f(k)\cdot n^{O(1)}$ time with $f$ being some computable
function. An \emph{fpt-reduction} from a parameterized problem $L_1 \subseteq \Sigma_1^* \times \mathbb{N}$ to $L_2 \subseteq \Sigma_2^* \times \mathbb{N}$ is a function $P: \Sigma_1^* \times \mathbb{N} \rightarrow \Sigma_2^* \times \mathbb{N}$ that preserves membership (i.e., $(x, k) \in L_1 \Leftrightarrow P((x, k)) \in L_2$), is computable in $f(k) \cdot |x|^{O(1)}$ time for some computable function $f$, and there exists a computable function $g$ such that for all  $(x,k) \in L_1$, if  $(x', k') = P((x, k))$, then $k' \leq g(k)$.
It is easy to verify that if $L_1$ and $L_2$ are parameterized
problems such that $L_1$ fpt-reduces to $L_2$ and $L_2$ is in \FPT, then
it follows that $L_1$ is in \FPT, too.
There are many parameterized classes with less desirable running times than \FPT but
we focus on \paraNP-hard problems: a problem is \paraNP-hard under fpt-reductions if it is \NP-hard for some constant parameter value, implying such problems are not in \FPT unless \Poly = \NP.

We continue by defining \emph{constraint satisfaction problems}. First, a \emph{constraint language} is a (typically finite) set of relations $\cA$ over a universe $A$, and for a relation $R \in \Gamma$ we write $\ar(R) = k$ to denote its arity $k$. It is sometimes convenient to associate a constraint language with a relational signature, and thus obtaining a \emph{relational structure}:
a tuple $(A; \tau, I)$ where
$A$ is the \emph{domain}, or \emph{universe},
$\tau$ is a relational signature, and $I$ is a function from $\sigma$ to the
set of all relations over $D$ which assigns each relation symbol $R$ a
corresponding relation $R^{\mathcal{A}}$ over $D$. We write $\ar(R)$ for the arity of a relation $R$, and if $R = \emptyset$ then $\ar(R) = 0$.
All structures in this paper are relational
and we assume that they have a finite signature unless otherwise stated.
Typically, we do not need to make a sharp distinction between relations and the corresponding relation symbols, so we usually simply write $(A; R_1, \ldots, R_m)$, where 
each $R_i$ is a relation over $A$, to denote a structure. 
We also sometimes do not make a sharp distinction between structures and sets of relations when the signature is not important. For arbitrary structures $\cA$ and $\cA'$ with domains $A$ and $A'$, we let
$\cA \cup \cA'$ denote the structure with domain $A \cup A'$ and containing the relations in $\cA$ and $\cA'$.

For a constraint language (or structure) $\cA$ an instance of the \emph{constraint satisfaction problem} over $\cA$ ($\csp{\cA}$) is then given by 
$I = (V,C)$ where $V$ is a set of variables and $C$ a set of constraints of the form $R(x_1, \ldots, x_k)$ where $x_1, \ldots, x_k \in V$ and $R \in \cA$, and the question is whether there exist a function $f \colon V \rightarrow A$ that satisfies all constraints (a \emph{solution}), i.e., $(f(x_1), \ldots, f(x_k)) \in R$ for all $R(x_1, \ldots, x_k) \in C$. 
The \emph{CSP dichotomy theorem} says that all finite-domain CSPs are either in P or are NP-complete~\cite{Bulatov:focs2017,Zhuk:jacm2020}.
Given an instance $I=(V,C)$ of $\csp{\cA}$, we let $\Sol(I)$ be the set of solutions to $I$. 
We now define CSPs with alien constraints in the style of Cohen et al.~\cite{Cohen:etal:jacm2000}.

\problemDef{$\alien{}{\cA}{\cB}$}
{A natural number $k$ and an instance $I = (V,C_1 \cup C_2)$ of CSP$(\cA \cup \cB$), where  $(V, C_1)$ is an instance of CSP$(\cA)$ and $(V,C_2)$ is an instance of CSP$(\cB)$ with $|C_2| \leq k$.}
{Does there exist a satisfying assignment to $I$?}

Throughout the paper, we assume without loss of generality that the structures $\cA$ and $\cB$ can be associated with disjoint signatures.
The parameter in $\alien{}{\cA}{\cB}$ is the \emph{number of alien constraints} (abbreviated \numac{}). 
We let
$\alien{k}{\cA}{\cB}$ denote the $\alien{}{\cA}{\cB}$ problem restricted
to a fixed value $k$ of parameter \numac{}.
Note that if CSP$(\cA)$ is not in \Poly, then  $\alien{0}{\cA}{\cB}$ is not in \Poly; moreover,
if CSP$(\cA \cup \cB)$ is in \Poly, then $\alien{}{\cA}{\cB}$ is in \Poly.
Thus, it is sensible to always require that CSP$(\cA)$ is in \Poly and CSP$(\cA \cup \cB)$
is not in \Poly. In many natural cases (e.g., all finite-domain CSPs),
CSP$(\cA \cup \cB)$ not being polynomial-time
solvable implies that CSP$(\cA \cup \cB)$ is \NP-hard.

A $k$-ary relation $R$ is said
to have a \emph{primitive positive definition} (pp-definition) over
a constraint language $\cA$ if $R(x_1, \ldots, x_{k}) \equiv \exists y_1, \ldots,
y_{k'}\, \colon \, R_1(\mathbf{x_1}) \wedge \ldots \wedge
R_m({\mathbf{x_m}})$ where each $R_i \in \cA \cup \{=_A\}$ and
each $\mathbf{x_i}$ is a tuple of variables over
$x_1,\ldots, x_{k}$, $y_1, \ldots, y_{k'}$ matching the arity of
$R_i$. Here, and in the sequel, $=_A$ is the equality relation over $A$, i.e. $\{(a,a) \mid a \in A\}$.  
If $\cA$ is a constraint
language, then we let $\cclone{\cA}$ be the inclusion-wise
smallest set of relations containing $\cA$ closed under
pp-definitions. 

\begin{theorem}[\cite{Jeavons:tcs98}] \label{thm:pp-red}
Let $\cA$ and $\cB$ be structures with the same domain. 
If every relation of $\cA$ has a primitive positive definition in $\cB$, then there is a polynomial-time reduction from CSP$(\cA)$ to CSP$(\cB)$.
\end{theorem}

When working with problems of the form $\alien{k}{\cA}{\cB}$ we additionally introduce the following simplifying notation:
$\cclone{\mathcal{A} \cup \mathcal{B}}_{\leq k}$ 
denotes the set of all pp-definable relations over $\mathcal{A} \cup \mathcal{B}$ using at most $k$ atoms from $\mathcal{B}$. 
We now describe the corresponding algebraic objects.
An operation 
$f \colon D^{m}\to D$ is a \emph{polymorphism} of a relation 
$R\subseteq D^{k}$ if, for any choice of $m$ tuples 
$(t_{{11}},\dotsc ,t_{{1k}}),\dotsc ,(t_{{m1}},\dotsc ,t_{{mk}})$ from $R$, it holds that 
$(f(t_{{11}},\dotsc ,t_{{m1}}),\dotsc ,f(t_{{1k}},\dotsc ,t_{{mk}}))$ 
is in $R$. 
An \emph{endomorphism} is a polymorphism with arity one.
If $f$ is a polymorphism of $R$, then we sometimes say that $R$ is \emph{invariant} under $f$.
A constraint language $\cA$ has the polymorphism $f$ if every relation in $\cA$ has $f$ as a polymorphism.
We let $\Pol(\cA)$ and $\End(\cA)$ denote the sets of polymorphisms and endomorphisms of $\cA$,
respectively. If $F$ is a set of functions
over $D$, then $\Inv(F)$ denotes the set of relations over $D$ that are invariant under every
function in $F$.
There are close algebraic connections between the operators $\cclone{\cdot}$, $\Pol(\cdot)$, and
$\Inv(\cdot)$. For instance, if $\cA$ has a finite domain (or, more generally, if $\cA$ is
$\omega$-categorical; see below), then we have a Galois connection $\cclone{\cA}=\Inv(\Pol(\cA))$~\cite[Theorem~5.1]{Bodirsky:Nesetril:csl2003}.

Polymorphisms enable us to compactly describe the tractable cases
of Boolean CSPs.

\begin{theorem}[\cite{Schaefer:stoc78}]
\label{thm:schaefer-result}
Let $\cA$ be a constraint language over the Boolean domain. The problem CSP($\cA$) is decidable in polynomial time if $\cA$ is invariant under one of the following six operations:
(1)
the constant unary operation 0 ($\cA$ is \textrm{0-valid}), 
(2)
the constant unary operation 1 ($\cA$ is \textrm{1-valid}), 
(3)
the binary min operation $\sqcap$ ($\cA$ is \textrm{Horn}), 
(4)
the binary max operation $\sqcup$ ($\cA$ is \textrm{anti-Horn}), 
(5)
the ternary majority operation 
$M(x,y,z)=(x\sqcap y)\sqcup (x\sqcap z)\sqcup (y\sqcap z)$ ($\cA$ is \textrm{2-SAT}), or
(6)
the ternary minority operation 
$m(x,y,z)=x \oplus y \oplus z$ where $\oplus$ is the addition operator in GF$(2)$ ($\cA$ is \textrm{affine}). 
Otherwise, the problem CSP$(\cA)$ is NP-complete.
\end{theorem}

A Boolean constraint language that satisfies condition (3), (4), (5), or (6) is called \emph{Schaefer}.

A finite-domain structure $\cA$ is a \emph{core} if every $e \in \End(\cA)$ is a bijection. We
let  $f(R) = \{(f(t_1), \ldots, f(t_n)) \mid (t_1, \ldots, t_n) \in R\}$ when
$f \colon A \rightarrow A$ and $R \in \cA$.
If $e \in \End(\cA)$ has minimal range, then $e(\cA)=\{e(R) \; | \; R \in \cA\}$ is a core and this core is unique up to isomorphism. We can thus speak about \emph{the core} $\cA^c$ of $\cA$.
It is easy to see that CSP$(\cA)$ and CSP$(\cA^c)$ are equivalent under polynomial-time reductions (indeed, even log-space reductions suffice). Another useful equivalence concerns constant relations.
Let $\cA^+$ denote the structure $\cA$ expanded by all unary singleton
relations $\{(a)\}$, $a \in A$. If $\cA$ is a core, then CSP$(\cA)$ and CSP$(\cA^+)$ are equivalent under polynomial-time reductions~\cite{DBLP:conf/dagstuhl/BartoKW17}.

We will frequently consider $\omega$-\emph{categorical} structures.
An \emph{automorphism} of a structure $\cA$ is a permutation $\alpha$ of its domain $A$ such that both $\alpha$ and its inverse are homomorphisms. The set of all automorphisms of a structure $\cA$ is denoted by $\Aut(\cA)$, and forms a group with respect to composition. 
The \emph{orbit} of $(a_1,\dots,a_n) \in A^n$ in $\Aut(\cA)$ is the set 
$\{(\alpha(a_1),\dots,\alpha(a_n)) \mid \alpha \in \Aut(\cA) \}.$ 
Let $\Orb(\cA)$ denote the set of orbits of $n$-tuples in $\Aut(\cA)$ (for all $n \geq 1$).
A structure $\cA$ with countable domain is $\omega$-categorical if and only if $\Aut(\cA)$ is \emph{oligomorphic}, i.e., it has only finitely many orbits of $n$-tuples for all $n \geq 1$. 

Two important classes of
$\omega$-categorical structures are \emph{equality languages} (respectively, \emph{temporal languages}) where each relation can be defined as the set of models of a first-order formula over $({\mathbb N};=)$ (respectively, $(\mathbb{Q}; <)$).  Importantly, $\Aut(\cA)$ is the full symmetric group if $\cA$ is an equality language. A relation in an equality language is said to be \emph{0-valid} if it contains \emph{any} constant tuple. This is justified since if the relation is invariant under one
constant operation, then it is invariant under all constant operations.%
The computational complexity of CSP for equality languages was classified by Bodirsky and 
Kára~\cite[Theorem~1]{Bodirsky:Kara:toct2008}: for any equality language $\cA$, CSP$(\cA)$ is solvable in
  polynomial time if $\cA$ is 0-valid or invariant under a binary injective operation, and is \NP-complete otherwise.

\section{Applications of Alien Constraints} \label{sec:alien-applications}

We will now demonstrate how alien constraints
can be used for studying the complexity of CSP-related problem:
Section~\ref{sec:redundant-etc} contains an example where
we analyse the complexity
of \emph{redundancy}, \emph{equivalence}, and \emph{implication} problems, and 
we consider connections between the model checking problem and
CSPs with alien constraints in Section~\ref{sec:modelchecking}.
To relate problem complexity we use \emph{Turing reductions}: a problem $L_1$ is \emph{polynomial-time Turing reducible} to $L_2$ (denoted $L_1 \leq^p_T L_2$) if it can be solved in polynomial time using an oracle for $L_2$. 
Two problems $L_1$ and $L_2$ are \emph{polynomial-time Turing equivalent} if $L_1 \leq^p_T L_2$ and $L_2 \leq^p_T L_1$.

\subsection{The Redundancy Problem and its Relatives}
\label{sec:redundant-etc}

We will now study the complexity of a family of well-known computational problems.
We begin by some definitions. Let $\cA$ denote a constraint language and assume that $I=(V,C)$ is an instance of CSP$(\cA)$. 
We say that a constraint $c \in C$ is \emph{redundant} in $I$
if $\Sol((V,C)) = \Sol((V,C \setminus \{c\}))$. 
We have the following computational problems.

\problemDef{\textsc{Redundant}$(\cA)$}
{An instance $(V,C)$ of CSP$(\cA)$ and a constraint $c \in C$.}
{Is $c$ redundant in $(V,C)$?}

\problemDef{\textsc{Impl}$(\cA)$}
{Two instances $(V,C_1), (V, C_2)$ of CSP$(\cA)$.}
{Does $(V, C_1)$ imply $(V, C_2)$, i.e., is it the case that $\mathrm{Sol}((V,C_1)) \subseteq \mathrm{Sol}((V, C_2))$?}

\problemDef{\textsc{Equiv}$(\cA)$}
{Two instances $(V,C_1), (V, C_2)$ of CSP$(\cA)$.}
{Is it the case that $\mathrm{Sol}((V,C_1)) = \mathrm{Sol}((V, C_2))$?}

Before we start working with alien constraints, we exhibit a
close connection between
$\textsc{Redundant}(\cdot)$,
\textsc{Equiv}$(\cdot)$, and $\textsc{Impl}(\cdot)$.

\begin{lemma} \label{lemma:equiveasy}
Let $\cA$ be a structure.
The problems \textsc{Equiv}$(\cA)$, $\textsc{Impl}(\cA)$, and
$\textsc{Redundant}(\cA)$ are polynomial-time Turing equivalent.
\end{lemma}
\begin{proof}
We show that (1) $\textsc{Equiv}(\cA) \leq_T^p \textsc{Impl}(\cA)$,
(2) $\textsc{Impl}(\cA) \leq_T^p \textsc{Redundant}(\cA)$, and
(3) $\textsc{Redundant}(\cA) \leq_T^p\textsc{Equiv}(\cA)$.

(1). Let $((V,C_1),(V,C_2))$ be an instance of \textsc{Equiv}$(\cA)$. We need to check whether $\Sol((V,C_1)) = \Sol((V, C_2))$. This is true if and only if the two $\textsc{Impl}$ instances $((V,C_1),(V,C_2))$
and $((V,C_2),(V,C_1))$ are yes-instances.

(2). Let $((V, C_1),(V, C_2))$ be an instance of \textsc{Impl}$(\cA)$. For each constraint $c \in C_2$, first check whether $C_1$ implies $\{c\}$ by (a) checking if $c \in C_1$, in which case $C_1$ trivially implies $\{c\}$, (b) if not, then check whether $c$ is redundant in $C_1 \cup \{c\}$, in which case we answer yes, and otherwise no. If $C_1$ implies  $\{c\}$ for every $c \in C_2$ then $C_1$ implies  $C_2$ and we answer yes, and otherwise no.

(3). Let $I=((V,C),c)$ be an instance of
\textsc{Redundant}$(\cA)$. It is obvious that $I$ is a yes-instance if and only
if the instance $((V,C),(V,C \setminus \{c\}))$ is a yes-instance of 
$\textsc{Equiv}(\cA)$.
\end{proof}  

Next, we show how the complexity of \textsc{Redundant}$(\cA)$ can be analysed by exploiting CSPs with alien constraints. 
If $R$ is a $k$-ary relation over domain $D$, then we let $\bar{R}$ denote its
\emph{complement}, i.e. $\bar{R}=D^k \setminus R$.

\begin{theorem} \label{thm:redundant-hybrid}
$(\star)$
Let $\cA$ be a structure with domain $A$.
If CSP$(\cA)$ is not in \Poly, then
\textsc{Redundant}$(\cA)$ is not in \Poly.
In particular,
\textsc{Redundant}$(\cA)$ is \NP-hard (under polynomial-time Turing reductions) whenever CSP$(\cA)$ is \NP-hard.
Otherwise,
\textsc{Redundant}$(\cA)$ is in \Poly if and only if
for every relation $R \in \cA$,
$\alien{1}{\cA}{\{\bar{R}\}}$ is in \Poly.
\end{theorem}

Combining Theorem~\ref{thm:redundant-hybrid} with the forthcoming
complexity classification of Boolean CSPs with alien constraints (Theorem~\ref{thm:boolean-classification})
shows that Boolean \textsc{Redundant}$(\cA)$ is in \Poly if and only if $\cA$ is Schaefer.
We have not found this result in the literature
but we view it as folklore since it follows
from other classification results (start from \cite{Boehler:etal:csl2002} or \cite{schnoor2008a} and transfer the results
to \textsc{Redundant}$(\cA)$ with the aid of Lemma~\ref{lemma:equiveasy}).
However, we claim that our proof is very different when compared to
the proofs in \cite{Boehler:etal:csl2002} and \cite{schnoor2008a}):
Böhler et al.\ use a lengthy case analysis while Schnoor \& Schnoor in addition uses the so-called weak base method, which scales poorly since not much is known about this construction for non-Boolean domains.
We do not claim that our proof is superior, but we do not see how to generalize
the classifications by Böhler et al.\ and Schnoor \& Schnoor to larger (in particular infinite) domains since they are fundamentally based on Post's classification of Boolean clones. Such a generalization, on the other hand, is indeed possible with our approach.
We demonstrate in Section~\ref{sec:equality-classification} that we can obtain a full understanding of the complexity
of CSPs with alien constraints for equality languages. This result carries over
to \textsc{Redundant}$(\cdot)$ via Theorem~\ref{thm:redundant-hybrid}, implying that we have a full complexity classification
of \textsc{Redundant}$(\cdot)$ for equality languages.
This result can immediately be transferred to \textsc{Impl}$(\cdot)$ and $\textsc{Equiv}(\cdot)$
by Lemma~\ref{lemma:equiveasy}.

\subsection{Model Checking}
\label{sec:modelchecking}

We follow \cite{Madelaine:Martin:sicomp2018} and view the \emph{model checking} problem as follows: given a logic $\mathscr{L}$, a structure $\cA$,
and a sentence $\phi$ of $\mathscr{L}$, decide whether $\cA \models \phi$.
The main motivation for this problem is its connection to databases~\cite{Vardi:stoc82}.
From the CSP perspective, we consider a slightly reformulated version: given an instance
$I=(V,C)$ of CSP$(\cA)$ and
a formula $\phi$ with free variables in $V$, we ask if there is a tuple in $\Sol(I)$ that satisfies $\phi$.
If $\phi$ can be expressed as an instance $I'$ of CSP$(\cB)$ for some
structure $\cB$, then
this is the same thing as if asking whether $I \cup I'$ has a solution 
or not. In the model-checking setting, we want to check whether $\phi$
is true in all solutions of $I$. If $\neg \phi$  can be expressed as an instance $I'$ of CSP$(\cB)$ for some
structure $\cB$, then we are done: every solution to $I$ satisfies
$\phi$ if and only if CSP$(I \cup I')$ is not satisfiable, and this clarifies
the connection with CSPs with alien constraints. For instance, one may view \textsc{Impl}$(\cA)$ (and consequently the underlying \alien{1}{\cA}{\bar{R}} problems by
Lemma~\ref{lemma:equiveasy} and Theorem~\ref{thm:redundant-hybrid})
as the model checking problem restricted to queries that are 
$\cA$-sentences constructed using the operators $\forall$ and $\vee$.
Naturally, one wants the ability to use more complex queries such as (1)
queries extended with other relations, i.e. queries constructed over an expanded
structure, or (2)
queries that are built using other logical connectives.

In both cases, it makes sense to study the fixed-parameter tractability of 
$\alien{}{\cA}{\cB}$ with parameter \numac{}
since the query is typically much smaller than the structure $\cA$.
The connection is quite obvious in the first case (one may view \numac{} as measuring how ``complex'' the given query is) while it is more hidden in the second case.
Let us therefore consider the negation operator.
From a logical perspective, one
may view a constraint $\bar{R}(x_1,\dots,x_k)$ as the formula $\neg R(x_1,\dots,x_k)$.
Needless to say, the relation $\bar{R}$ is often not pp-definable in a structure $\cA$
containing $R$ but it may be existential positive definable in $\cA$. 
Assume that the preconditions of the example hold and that $\csp{\cA}$ is
in \Poly. We know that $\bar{R}$ has an existential positive definition in $\cA$
for every $R \in \cA$.
Let $\bar{\cA} = \{\bar{R} \mid R \in \cA\}$ and consider the problem
$\alien{}{\cA}{\bar{\cA}}$. The forthcoming Theorem~\ref{thm:fpt-result-variant} is applicable so this problem is in \FPT parameterized by $\#{\sf ac}$.
Now, the corresponding model checking problem is to decide if $\cA \models \phi$ where
$\phi$ is an $\cA$-sentence constructed using the operators $\forall$ and $\vee$ 
and where we are additionally allowed to use negated relations $\neg R(x_1,\dots,x_m)$. It follows that this problem is in \FPT parameterized by the number of negated relations.

\section{General Tools for Alien Constraints}
\label{sec:generalresults}

We analyze the complexity of $\alien{k}{\cA}{\cB}$, starting in Section~\ref{sec:algebra} by investigating which of the classic algebraic tools are applicable to the alien constraint setting, and continuing in Section~\ref{sec:general-fpt} by presenting a general FPT result.
We will use these observations for proving various results but also for obtaining a better understanding of alien constraints.

\subsection{Alien Constraints and Algebra} \label{sec:algebra}

First, we have a straightforward generalization of Theorem~\ref{thm:pp-red} in the alien constraint setting.

\begin{theorem}
\label{thm:generalized-ppdefs}
$(\star)$ Let $\cA$ and $\cB$ be two structures with disjoint signatures. There exists
a polynomial time many-one reduction $f$ from
       $\alien{}{\cA^*}{\cB^*}$ to $\alien{}{\cA}{\cB}$ for any finite
       $\cA^* \subseteq \cclone{\cA}$ and
       $\cB^* \subseteq \cclone{\cA \cup \cB}$.
If $I=(V,C,k)$ is an instance of $\alien{}{\cA^*}{\cB^*}$ and
$f(I)=(V',C',k')$, then $k'$ only depends on $k$, $\cA$, $\cB$, and $\cB^*$, so $f$ is an fpt-reduction.
\end{theorem}

This claim is, naturally, in general not true for  $\alien{k}{\cA^*}{\cB}$ for finite $\cA^* \subseteq \cclone{\cA \cup \cB}$. 
The idea underlying Theorem~\ref{thm:generalized-ppdefs}
can be used in many different ways and we give one example.

\begin{proposition} \label{prop:k=1-red}
  If $\cA,\cB$ are structures and
  $R \in \langle \cA \cup \cB \rangle_{\leq 1}$, then
  $\alien{k}{\cA}{(\cB \cup \{R\})}$ is polynomial-time reducible
  to $\alien{k}{\cA}{\cB}$.
\end{proposition}

We proceed by relating $\alien{k}{\cA}{\cB}$ to the important idea of reducing to a core (recall Section~\ref{sec:prelims}). 
Recall that $\cA^c$ denotes the (unique up to isomorphism) core of a finite-domain structure $\cA$. For two structures $\cA \cup \cB$ we similarly write $(\cA \cup \cB)^c$ for the core. Specifically, if $e \in \End(\cA \cup \cB)$ has minimal range, then the core consists of $\{e(R) \mid R \in \cA\} \cup \{e(R) \mid R \in \cB\}$ of the same signature as $\cA$ and $\cB$, and the problem $\balien{}{(\cA \cup \cB)^c}$ is thus well-defined.

\begin{theorem} \label{thm:core-reduction}
$(\star)$
Let $\cA$ and $\cB$ be two structures over a finite universe $A$. Then $\alien{}{\cA}{\cB}$ and $\balien{}{(\cA \cup \cB)^c}$ are interreducible under both polynomial-time 
and fpt reductions.
\end{theorem}

In general, it is \emph{not} possible to reduce from $\alien{k}{\cA}{\cB}$  to 
${\alien{k}{\cA^c}{\cB}}$ or from ${\alien{k}{\cA}{\cB}}$ to ${\alien{k}{\cA}{\cB^c}}$. 
This can be seen as follows. 
Consider the Boolean relation
$R(x_1,x_2,x_3) \equiv x_1 = x_2 \lor x_2 = x_3$, 
and let $\cA = \{R\}$, $\cB = \{\neq\}$.
Then, $\alien{1}{\cA}{\cB}$ is \NP-hard
(see e.g. Exercise~3.24~in~\cite{chen2006rendezvous})
so $\alien{}{\cA}{\cB}$ is \paraNP-hard.
However, $\cA$ is $0$-valid, so $\cA^c = \{\{(0,0,0)\}\}$,
implying that $\alien{}{\cA^c}{\cB}$ is in \Poly.

\subsection{Fixed-Parameter Tractability} 
\label{sec:general-fpt}

We present an algorithm in this section that underlies many of our fixed-parameter tractability results and it is based on a particular notion of definability.
The \emph{existential
fragment} of first-order logic consists of formulas that only use the operations negation, conjunction, disjunction, and
existential quantification, while the \emph{existential positive} fragment additionally
disallows negation. We emphasize that it is required that the equality relation is
allowed in existential (positive) definitions. 
We can view existential positive in a different way that is easier to use
in our algorithm.
Let $\cA$ be a structure
with domain $A$ and
assume that $R \subseteq A^m$ is defined via a
existential positive definition over $\cA$, i.e., 
$R(x_1,\dots,x_m) \equiv
\exists y_1,\dots,y_n \colon \;
\phi(x_1,\dots,x_m,y_1,\dots,y_n)$
where
$\phi$ is a quantifier-free existential positive $\cA$-formula.
Since $\phi$ can be written in disjunctive normal form without introducing negation or
quantifiers, it follows
that $R$ is a finite union of relations in $\langle \cA \rangle$.

\begin{theorem} \label{thm:fpt-result-variant2}
Assume the following.
\begin{enumerate}
\item
$\cA,\cB$ are structures
with the same domain $A$,

\item
every relation in $\cB$ is existential positive definable in
$\cA$,
and

\item
CSP$(\cA)$ is in P.
\end{enumerate}
Then $\alien{}{\cA}{\cB}$ is in \FPT parameterized by \numac{}.
\end{theorem}
\begin{proof}
Assume $\cB = \{A;B_1,\dots,B_m\}$.
Condition 2.\ implies that $B_i$, $i \in [m]$, is a
finite union of relations $B_i=R^1_i \cup \dots \cup R^{c_i}_i$
where $R^1_i,\dots,R^{c_i}_i$ are in $\langle \cA \rangle$.
Let the structure $\cA^*$ contain the relations in $\cA \cup \{R_i^j \; | \; i \in [m] \; \textrm{and} \; j \in [c_i]\}$.
Clearly, $\cA^*$ has a finite signature and the problem CSP$(\cA^*)$ is
in \Poly by Theorem~\ref{thm:pp-red} since every relation in $\cA^*$ is a member of $\langle \cA \rangle$. Let $b=\max \{c_i \; | \; i \in [m]\}$.

Let $((V,C),k)$ denote an arbitrary instance of
$\alien{}{\cA}{\cB})$.
The satisfiability of $(V,C)$ can be
checked via the following procedure.
If $C$ contains no $\cB$-constraint, then check the satisfiability
of $(V,C)$ with the polynomial-time algorithm for CSP$(\cA)$.
Otherwise,
pick one constraint $c=B_i(x_1,\dots,x_q)$ with $B_i \in \cB$ and
check recursively the satisfiability of the following
instances: \[(V,(C \setminus \{c\}) \cup \{R_{i}^1(x_1,\dots,x_q)\}),\dots,(V,(C \setminus \{c\}) \cup \{R_{i}^{c_i}(x_1,\dots,x_q)\}).\]
If at least one of the instances
is satisfiable, then answer "yes" and otherwise "no". This is clearly a correct algorithm for $\alien{}{\cA}{\cB}$.

We continue with the complexity analysis. Note that the
leaves in the computation tree produced by the algorithm are CSP$(\cA^*)$ instances and they are consequently
solvable in polynomial time.
The depth of the computation tree is at most $k$ (since $(V,C)$ contains at most
$k$ $\cB$-constraints) and each node
has at most $b$ children. Thus, the problem can be solved in
$b^k \cdot \textrm{poly}(|I|)$ time.
We conclude that $\alien{}{\cA}{\cB}$ is in \FPT parameterized by \numac{} since $b$ is a fixed constant that only depends on
the structures $\cA$ and $\cB$.
\end{proof}

\section{Finite-Domain Languages}
\label{sec:findom-languages}

This section is devoted to CSPs over finite domains. 
We begin in Section~\ref{sec:constants-finite} by studying how the definability
of constants affect the complexity of finite-domain CSPs with alien constraints, 
and we use this as a cornerstone for a parameterized \FPT versus \paraNP
dichotomy result for
of $\alien{}{\cA}{\cB}$. We show a sharper result for Boolean structures
in Section~\ref{sec:class-boolean}.

\subsection{Parameterized Dichotomy}
\label{sec:constants-finite}
 
We begin with a simplifying result.
For a finite set $A$, let $\cC_A$ be the structure
whose relations are the constants over $A$.

\begin{lemma} \label{lem:constant-reduction-stronger}
$(\star)$    
  Let $\mathcal{A}$ be a structure over a domain $A$. For every $\mathcal{C} \subseteq \mathcal{C}_A$,
  CSP$(\mathcal{A} \cup \mathcal{C})$ is polynomial-time reducible to $\alien{|\mathcal{C}|}{\cA}{\cC}$.
\end{lemma}

Lemma~\ref{lem:constant-reduction-stronger} together with the basic algebraic results from Section~\ref{sec:algebra} allows us to prove the following result
that combines a more easily formulated fixed-parameter result (compared to
Theorem~\ref{thm:fpt-result-variant2}) with a powerful
hardness result. 

\begin{theorem} \label{thm:findomcsp}
$(\star)$
  Let $\cA,\cB$ be structures
  with finite domain $D$.
  Assume that $\cA 
  \cup \cB$ is a core.
  If CSP$(\cA \cup \cC_A)$ is in \Poly, then
  $\alien{}{\cA}{\cB}$ is in \FPT with
  parameter \numac{}. Otherwise,
  $\alien{p}{\cA}{\cB}$ is \NP-hard for some
  $p$ that only depends on the structures $\cA$ and $\cB$.
\end{theorem}

\begin{proof}
    We provide a short proof sketch, the full proof is in Appendix~\ref{sec:proof_of_thm:findomcsp}.
    Using the dichotomy of finite domain CSPs~\cite{Bulatov:focs2017,Zhuk:jacm2020}, we first assume CSP$(\cA \cup \cC_D)$ is in \Poly. One can prove that every tuple over $D$ is pp-definable over $\cA \cup \cC_D$ and then that each relation in $\cB$ is existential positive definable over $\cA \cup \cC_D$. We can now apply Theorem~\ref{thm:fpt-result-variant2}, and $\alien{}{\cA}{\cB}$ is in \FPT.

For the \NP-hard case, we assume CSP$(\cA \cup \cC_D)$ is \NP-hard and construct a polynomial-time reduction from CSP$(\cA \cup \cC_D)$ to $\alien{p}{\cA}{\cB}$. We use the endomorphisms of $\cA \cup \cB$ to construct a pp-definable relation $E$ which allow us to simulate the constant relations, and a reduction to $\alien{1}{\cA}{\{E\}}$ to establish the claim via Lemma~\ref{lem:constant-reduction-stronger} and Theorem~\ref{thm:generalized-ppdefs}.
\end{proof}

Theorem~\ref{thm:findomcsp} has broad applicability.
Let us, for instance, consider a structure $\cA$ with finite domain $A$ and containing a finite number of relations from $\Inv(f)$ where
$f \colon A^m \rightarrow A$ is idempotent ($f:A^m \rightarrow D$ is \emph{idempotent} if
$f(a,\dots,a)=a$ for all $a \in A$.)
If $\csp{\cA}$ is in \Poly, then $\csp{\cA \cup \cC_A}$ is in \Poly since constant
relations are invariant under $f$. 
Hence, $\alien{}{\cA}{\cB}$ is in \FPT parameterized by \numac{} for \emph{every}
finite structure $\cB$ with domain $A$ by Theorem~\ref{thm:findomcsp}.
Idempotent functions that give rise to polynomial-time solvable
CSPs are fundamental and well-studied in the literature; see e.g.\ the survey by Barto et al.~\cite{DBLP:conf/dagstuhl/BartoKW17}.

Via Theorem~\ref{thm:core-reduction} we obtain the following parameterized complexity dichotomy separating problems in \FPT from \paraNP-hard problems.

\begin{corollary}
\label{cor:param-dichotomy}
Let $\cA,\cB$ be structures
over the finite domain $A$.
Then, $\alien{}{\cA}{\cB}$ is either in \FPT or \paraNP-hard (in parameter \numac{}).
\end{corollary}
\begin{proof}
Let $e \in \End(\cA \cup \cB)$ have minimal range and let $\cA' = \{e(R) \mid R \in \cA\}$ and $\cB' = \{R \mid R \in \cB\}$ be the two components of the core $(\cA \cup \cB)^c$, and let  $A' = \{e(a) \mid a \in A\}$ be the resulting domain.
The problems $\alien{}{\cA}{\cB}$ and
$\alien{}{\cA'}{\cB'}$ are fpt-interreducible by Theorem~\ref{thm:core-reduction}.
The problem $\csp{\cA' \cup  \cC_{A'}}$ is either in \Poly or is \NP-hard by the CSP dichotomy theorem~\cite{Bulatov:focs2017,Zhuk:jacm2020}.
In the first case,  
$\alien{}{\cA'}{\cB'}$ (and thus $\alien{}{\cA}{\cB}$)
is in \FPT with parameter \numac{}.
Otherwise, $\alien{}{\cA'}{\cB'}$ is \paraNP-hard, and 
the fpt-reduction from $\alien{}{\cA'}{\cB'}$ to $\alien{}{\cA}{\cB}$ 
establishes \paraNP-hardness for the latter.
\end{proof}

Corollary~\ref{cor:param-dichotomy} must be used with caution:
it does not imply that $\alien{1}{\cA}{\cB}$
is \NP-hard and results such as
Theorem~\ref{thm:redundant-hybrid}
may not be applicable.
This encourages the refinement of
coarse complexity results based on Theorem~\ref{thm:findomcsp}.
We use Boolean relations as an example of this in the next section.

\subsection{Classification of Boolean Languages}
\label{sec:class-boolean}

We present a complexity classification of
$\alien{}{\cA}{\cB}$ when $\cA$ and $\cB$ are Boolean structures
(Theorem~\ref{thm:boolean-classification}).
We begin with two auxiliary results and we define relations $c_0 = \{(0)\}$ and $c_1 = \{(1)\}$.

\begin{lemma} \label{lem:one-occurrence-pp} 
$(\star)$ Let $\cA$ be a Boolean structure where $c_0 \in \cclone{\cA}$. If an $n$-ary Boolean $R \neq \emptyset$ is not 0-valid then $c_1 \in \cclone{\mathcal{A} \cup \{R\}}_{\leq 1}$.
\end{lemma}

We say that a Boolean relation $R$ is \emph{invariant under complement} if
it is invariant under the operation $\{0 \mapsto 1,1 \mapsto 0\}$.
This is equivalent to
$(t_1,\dots,t_k) \in R$ if and only if
$(1-t_1,\dots,1-t_k) \in R$.

\begin{lemma} \label{lem:neq-reduction}
$(\star)$ Let $\cA$ be a Boolean structure with finite signature.
If $\cA$ is invariant under complement,
then CSP$(\cA \cup \{c_0,c_1\})$ is polynomial-time
reducible to 
$\alien{1}{\cA}{\{\neq\}}$.
\end{lemma}

We are now ready for analysing the complexity of $\alien{}{\cA}{\cB}$
when $\cA$ and $\cB$ are Boolean structures. 
We use a simplifying concept:
a \emph{0/1}-pair $(R_0,R_1)$ contains two Boolean relations
where $R_0$ is 0-valid but not 1-valid and $R_1$ is 1-valid but not
0-valid.

\begin{theorem} \label{thm:boolean-classification}
Let $\cA$ and $\cB$ be Boolean structures
 such that
CSP$(\cA)$ is in \Poly and CSP$(\cA \cup \cB)$
is \NP-hard. Then the following holds.

\begin{enumerate}

\item
If $\cA$ is Schaefer, then $\alien{}{\cA}{\cB}$ is in \FPT with parameter \numac{}. 

\item
If (i) $\cA$ is not Schaefer, (ii) $\cA$ is
both 0- and 1-valid, (iii) $\cB$ contains a $0/1$-pair, and (iv) 
$\cB$ is 0- or 1-valid, then 
$\alien{2}{\cA}{\cB}$  is \NP-hard and
$\alien{1}{\cA}{\cB}$  is in \Poly.

\item
Otherwise, $\alien{1}{\cA}{\cB}$ is 
\NP-hard.
\end{enumerate}
\end{theorem}
\begin{proof}
Assume $\cA$ is Schaefer and let $\cA^+=\cA \cup \{c_0,c_1\}$.
The structure $\cA^+$ is clearly a core and
$\cA^+ \cup \cB$ is a core, too.
The problem
CSP$(\cA^+)$ is
in \Poly by Theorem~\ref{thm:schaefer-result} so Theorem~\ref{thm:findomcsp} implies that
$\alien{}{\cA^+}{\cB}$ (and naturally $\alien{}{\cA}{\cB}$) is in \FPT parameterized by \numac{}.
Since CSP$(\cA)$ is in \Poly, we know from Theorem~\ref{thm:schaefer-result}
that
$\cA$ is 0-valid, 1-valid or Schaefer.
We assume henceforth that $\cA$ is 0-valid and not Schaefer; the other case is analogous.
If $\cB$ is 0-valid, then CSP$(\cA \cup \cB)$ is trivially in \Poly and this is ruled out by our initial assumptions. We assume henceforth that $\cB$ is not 0-valid
and consider two cases depending on whether $c_0$ is pp-definable
in $\cA$ or not.

\smallskip

\noindent
\emph{Case 1.}
$c_0$ is pp-definable in $\cA$.
We know that CSP$(\cA \cup \{c_0,c_1\})$ is \NP-hard by Theorem~\ref{thm:schaefer-result} since $\cA$ is not Schaefer.
We can thus assume that CSP$(\cA \cup \{c_1\})$ 
is \NP-hard.
Lemma~\ref{lem:constant-reduction-stronger} implies that
CSP$_{\leq 1}(\cA \cup \{c_1\})$ is \NP-hard.
The relation $c_1$ is in
$\langle \cA \cup \cB \rangle_{ \cB \leq 1}$
by Lemma~\ref{lem:one-occurrence-pp} so
we conclude that 
CSP$_{\leq 1}(\cA \cup \cB)$ is \NP-hard.

\smallskip

\noindent
\emph{Case 2.} $c_0$ is not pp-definable in $\cA$.
This implies that every relation in $\cA$ is simultaneously
0- and 1-valid. To see this, assume to the contrary that
$\cA$ contains a relation that is not 1-valid.
Then, $x=0 \Leftrightarrow R(x,\dots,x)$ and $c_0$ is
pp-definable in $\cA$. This implies that $\cB$ contains
(a) a relation that is not invariant under any constant operation or (b)
every relation is closed under a constant operation and $\cB$
contains a $0/1$-pair. Note that if (a) and (b) does not hold, then
$\cB$ is invariant under a constant operation and
CSP$(\cA \cup \cB)$ is trivially in \Poly.

\smallskip

\noindent
\emph{Case 2(a).}
There is a
a relation $R$ in $\cB$ that is not invariant under any constant operation, i.e. $(0,\dots,0) \not\in R$
and $(1,\dots,1) \not\in R$. The relation $R$ has arity $a \geq 2$.
Let $t$ be the tuple in $R$ that contains the maximal number $b$
of 0:s. Clearly, $b < a$. We assume that the arguments are permuted
so that $t$ begins with $b$ 0:s and continues with $a-b$ 1:s.
Consider the pp-defintion
\[S(x,y) \equiv R(\underbrace{x,\dots,x}_{b \; \textrm{occ.}},\underbrace{y,\dots,y}_{a-b \; \textrm{occ.}}).\]
There are two possibilities: either $S(x,y) \Leftrightarrow x=0 \wedge y=1$
or $S(x,y) \Leftrightarrow x \neq y$. In the first case we are done
since CSP$(\cA \cup \{c_0,c_1\})$ is \NP-hard (recall that
$\cA$ is not Schaefer) and
$\alien{1}{\cA}{\cB}$  is easily seen to be
\NP-hard by Lemma~\ref{lem:constant-reduction-stronger}.
Let us consider the second case.
If $\cA$ is invariant under complement, then 
$\alien{1}{\cA}{\cB}$  is \NP-hard
by Lemma~\ref{lem:neq-reduction}.
If $\cA$ is not invariant under complement, then we claim that $c_0$ and
$c_1$ can be pp-defined with the aid of $\neq$.
Arbitrarily choose a relation $T$ in $\cA$ that contains
a tuple $t=(t_1,\dots,t_a)$ such that
$(1-t_1,\dots,1-t_a) \not\in T$---note that $t$ cannot be a constant tuple
since both $(0,\dots,0)$ and $(1,\dots,1)$ are in $T$.
Assume that $t$ contains $b$ 0:s and that the arguments
are permuted so that $t$ begins with $b$ 0:s followed by
$a-b$ 1:s.
Consider the pp-definition
\[U(x,y) \equiv x \neq y \wedge T(\underbrace{x,\dots,x}_{b \; \textrm{occ.}},\underbrace{y,\dots,y}_{a-b \; \textrm{occ.}}).\]
The relation $U$ contains the single tuple $(0,1)$. We know that
CSP$(\cA \cup \{c_0,c_1\})$ is \NP-hard (recall that
$\cA$ is not Schaefer) and
Lemma~\ref{lem:constant-reduction-stronger} implies that
$\alien{2}{\cA}{\{c_0,c_1\}}$ is \NP-hard, too.
It is now easy to see that
$\alien{1}{\cA}{\cB}$  is \NP-hard via the definition of $U$.

\smallskip

\noindent
\emph{Case 2(b).}
Every relation in $\cB$ is closed under at least one
constant operation and $\cB$ contains a 0/1-pair $(R_0,R_1)$.
Since $\cA$ is both 0- and 1-valid,
it follows that
$\alien{1}{\cA}{\cB}$ 
is in \Poly.
The constant relations $c_0$ and $c_1$ are pp-definable
in $\{R_0,R_1\}$ since $x=0 \Leftrightarrow R_0(x,\dots,x)$ and
$x=1 \Leftrightarrow R_1(x,\dots,x)$.
This implies with the aid of Lemma~\ref{lem:constant-reduction-stronger}
that 
$\alien{2}{\cA}{\cB}$ is \NP-hard
since $\cA$ is not Schaefer.
\end{proof}

Theorem~\ref{thm:boolean-classification} carries over to Boolean \textsc{Redundant}$(\cdot)$,
\textsc{Equiv}$(\cdot)$
and \textsc{Impl}$(\cdot)$ 
by Lemma~\ref{lemma:equiveasy} combined with Theorem~\ref{thm:redundant-hybrid}, so
these problems are in \Poly
if and only if $\cA$ is Schaefer (case 2. in Theorem~\ref{thm:boolean-classification}
is not applicable when analysing these problems since
it requires $|\cB| \geq 2$). Otherwise, they are \NP-complete under polynomial-time Turing reductions.
The \emph{meta-problem} for Boolean CSPs with alien constraints is decidable,
i.e., there is an algorithm that decides for Boolean structures $\cA,\cB$
whether $\alien{}{\cA}{\cB}$ is in case 1., 2., or 3. of Theorem~\ref{thm:boolean-classification}.
This is obvious since we have polymorphism descriptions of the Schaefer languages.

\section{Infinite-Domain Languages}
\label{sec:infindom-languages}

We focus on infinite-domain CSPs in this section.
We begin Section~\ref{sec:infdom-general} by 
discussing certain problems when CSPs with alien constraints are generalized to infinite domains. Our conclusion is that restricting ourselves
to $\omega$-categorical structures is a viable first step: $\omega$-categorical
structures constitute a rich class of CSPs and we can generalize at least some of the machinery from Section~\ref{sec:findom-languages} to this setting.
We demonstrate this in Section~\ref{sec:equality-classification} where we obtain
a complete complexity classification for equality languages.

\subsection{Orbits and Infinite-Domain CSPs}
\label{sec:infdom-general}

It is not straightforward to tranfer the results in Section~\ref{sec:findom-languages} to the
infinite-domain regime.
First, let us consider Theorem~\ref{thm:fpt-result-variant2}. In contrast to finite domains,
relations in $\cB$ may not be finite unions of relations in
$\langle \cA \rangle$ or, equivalently, not being definable with an existential positive formula. 
Second, let us consider Theorem~\ref{thm:findomcsp}:
the proof is based on structures expanded with symbols for each
domain value and this leads to problematic structures with infinite signatures.
The proof is also based on the assumption that CSPs are either polynomial-time
solvable or \NP-complete, and this is no longer true~\cite{Bodirsky:Grohe:icalp2008}.
It is thus necessary to restrict our attention to some class
of structures with sufficiently pleasant properties.
A natural choice is
$\omega$-categorical
structures that
allows us to reformulate 
Theorem~\ref{thm:fpt-result-variant2} as follows.

\begin{theorem} \label{thm:fpt-result-variant}
$(\star)$
Assume the following.
\begin{enumerate}
\item
$\cA,\cB$ are structures
with the same countable (not necessarily infinite) domain $A$,

\item
$\cA$ and $\cB$ are $\omega$-categorical, 

\item
every relation in $\Orb(\cB)$ is existential primitive definable in $\langle \cA \rangle$, and

\item
CSP$(\cA)$ is in \Poly
\end{enumerate}
Then $\alien{}{\cA}{\cB}$ is in \FPT parameterized by \numac{}.
\end{theorem}

\begin{example}
Results related to Theorem~\ref{thm:fpt-result-variant} have been presented in the
literature.
Recall that RCC5 and RCC8 are spatial formalism with binary relations 
that are disjunctions
of certain basic relations~\cite{Randell:etal:kr92}. 
Li et al.~\cite{Li:etal:ai2015} prove that if $\cA$ is a polynomial-time
solvable RCC5 or RCC8 constraint language containing all basic relations, then  \textsc{Redundant}$(\cA)$ is in \Poly. This immediately follows from combining
Theorem~\ref{thm:redundant-hybrid} and
Theorem~\ref{thm:fpt-result-variant} since
RCC5 and RCC8 can be represented by $\omega$-categorical
constraint languages~\cite{Bodirsky:Chen:jlc2009,Bodirsky:Woelfl:ijcai2011}
and every RCC5/RCC8 relation is existential primitive definable in the structure
of basic relations
by definition. This result can be generalized to a much larger class of relations in the case of RCC5 since the orbits of $k$-tuples are pp-definable in the structure of basic relations~\cite[Proposition 35]{Bodirsky:Jonsson:jair2017}.
\end{example}

A general hardness result based on the principles behind Theorem~\ref{thm:findomcsp} does not
seem possible in the infinite-domain setting, even for $\omega$-categorical structures. 
The hardness proof in Theorem~\ref{thm:findomcsp}
utilizes variables given fixed values and a direct generalization
would lead to
groups of variables
that together form an orbit of an $n$-tuple. Such gadgets behave very differently
from variables given fixed values: in particular, they do not admit a result similar
to Lemma~\ref{lem:constant-reduction-stronger}. Thus, hardness results needs to be
constructed in other ways.

We know from Section~\ref{sec:algebra} that $\alien{}{\cA}{\cB}$ and $\balien{}{(\cA \cup \cB)^c}$ are the same
when $\cA$ and $\cB$ has the same finite domain.
We now consider a generalisation of cores to infinite domains from Bodirsky~ \cite{Bodirsky:Book}:
an $\omega$-categorical structure $\cA$ with countable domain
is a \emph{model-complete core} if every relation in $\Orb(\cA)$ is pp-definable in $\cA$. 
There is an obvious infinite-domain analogue of Theorem~\ref{thm:core-reduction}: 
if $\cA' \cup \cB'$ is the model-complete core
of $\cA \cup \cB$ (where $\cA,\cB$ are $\omega$-categorical structures over
a countable domain $A$), then
$\alien{}{\cA}{\cB}$ polynomial-time reduces to $\alien{}{\cA'}{\cB'}$.
Model-complete cores share many other properties with cores, too. 
With this said, it is interesting to understand model-complete cores
in the context of $\alien{}{\cA}{\cB}$, simply because they are so well-studied
and exhibit useful properties. We merely touch upon this subject by making
an observation that we use in Section~\ref{sec:equality-classification}.

\begin{lemma} \label{lem:modelcompletefpt}
$(\star)$
Let $\cA$ and $\cB$ denote $\omega$-categorical structures with a countable domain $A$. Assume that $\cA$ is a
model-complete core and CSP$(\cA)$ is in \Poly. Then, $\alien{}{\cA}{\cB}$
is in \FPT parameterized by \numac{} for every structure $\cB$ such that $\Orb(\cB) \subseteq \Orb(\cA)$. 
\end{lemma}

\subsection{Classification of Equality Languages}
\label{sec:equality-classification}

We present a complexity classification
of $\alien{}{\cA}{\cB}$
for equality languages
$\cA$, $\cB$.
Essentially, there are two interesting cases:
when $\cA$ is Horn, and when $\cA$ is $0$-valid and not Horn.
In the former case, $\alien{}{\cA}{\cB}$
is in \FPT parameterized by \numac{},
while in the second case it is \paraNP-hard.
It turns out that the ability to
pp-define the arity-$c$ disequality relation,
where $c$ depends only on $\cA$,
using at most $k$ alien constraints,
determines the complexity.
A dichotomy for
$\textsc{Redundant}(\cdot)$,
$\textsc{Impl}(\cdot)$, and
$\textsc{Equiv}(\cdot)$ follows: these problems are either in \Poly or 
\NP-hard under polynomial-time Turing reductions.

Recall that
$\csp{\cA}$ for a finite equality constraint language $\cA$ is in \Poly if
$\cA$ is 0-valid or preserved by a binary injective operation,
and \NP-hard otherwise,
and that the automorphism group for equality languages is the symmetric group $\Sigma$
on $\Nat$, i.e. the set of permutations on $\Nat$.
It is easy to see that an orbit of a $k$-tuple $(a_1,\dots,a_k)$ is
pp-definable in $\{=,\neq\}$. For instance, the orbit of $(0,0,1,2)$
is defined by
$O(x_1,x_2,x_3,x_4) \equiv x_1 = x_2 \wedge x_2 \neq x_3 \wedge x_2 \neq x_4 \wedge x_3 \neq x_4$.
Observe that $\neq$ is invariant under every binary injective operation, so if $\cA$ is Horn, then $\neq \in \langle \cA \rangle$ and every orbit of $n$-tuples under $\Sigma$ is pp-definable in $\cA$. Thus,
$\cA$ is a model-complete core
as pointed out in Section~\ref{sec:infdom-general}.
Lemma~\ref{lem:modelcompletefpt} now implies the following.

\begin{corollary}
  \label{cor:Horn-fpt}
  Let $\cA$ and $\cB$ be equality languages.
  If $\cA$ is Horn,
  then $\alien{}{\cA}{\cB}$ is in \FPT parameterized by \numac{}.
\end{corollary}

Thus, we need to classify the complexity of
$\alien{k}{\cA}{\cB}$ for every $k$,
where $\cA$ is $0$-valid and not Horn,
and $\cB$ is not $0$-valid.
We will rely on results about the complexity of
singleton expansions of equality constraint languages.
Let $\cA$ be a constraint language over the domain $\Nat$.
By $\cA^+_c$ we denote the expansion of $\cA$ with $c$ singleton
relations, i.e. $\cA^+_c = \cA \cup \{ \{1\}, \dots, \{c\} \}$.
The complexity of $\csp{\cA^+_c}$ for equality constraint
languages $\cA$ and all constants $c$ was classified by Osipov \& Wahlstr\"om~\cite[Section~7]{equalityMinCSParXiv},
building on the detailed study of polymorphisms of
equality constraint languages 
by Bodirsky et al.~\cite{bodirsky_chen_pinsker_2010}.

The connection between $\alien{k}{\cA}{\cB}$ 
and $\csp{\cA^+_c}$ is the following.
In one direction,
we can augment every instance of $\csp{\cA}$
with $c$ fresh variables $z_1,\dots,z_c$ and, 
assuming $k$ is large enough
and $\cB$ is not $0$-valid,
use $\cB$-constraints to ensure that
$z_1,\dots,z_c$ attain distinct
values in every satisfying assignment.
Given that $\cA$ is invariant under 
every permutation of $\Nat$, 
we can now treat $z_1, \dots, z_c$ as constants,
e.g. as $1, \dots, c$, 
and transfer hardness results
from the singleton expansion to our problem.
In the other direction, if the relation
$\NEQQ{c+1} \notin \cclone{\cA \cup \cB}_{\leq k}$,
then every satisfiable instance of $\alien{k}{\cA}{\cB}$
has a solution with range $[c]$,
and $\cA^+_c$ is tractable:
indeed, a satisfiable instance without such
a solution would be a pp-definition of $\NEQQ{c'}$
for some $c' > c$.
These connections are formalized in
Lemmas~\ref{lem:neq-definition}~and~\ref{lem:expansion-to-alien}.
We will leverage the following hardness result.

\begin{lemma}[Follows from Theorem~54~in~\cite{equalityMinCSParXiv}]
  \label{lem:not-horn-expansion-nph}
  Let $\cA$ be a finite equality language.
  If $\cA$ is not Horn, then $\csp{\cA^+_c}$ is \NP-hard for some $c = c(\cA)$.
\end{lemma}

Our main tool for studying singleton expansions are \emph{retractions}.

\begin{definition}
  \label{def:retraction}
  Let $\cA$ be an equality language.
  An operation $f \colon \Nat \to [c]$ is a 
  \emph{retraction of $\cA$ to $[c]$}
  if $f$ is an endomorphism of $\cA$
  where
  $f(i) = i$ for all $i \in [c]$.
  If $\cA$ admits a retraction $f$ to $[c]$, then
  we say that \emph{$\cA$ retracts to $[c]$},
  and $\cA_f$ is \emph{a retract (of $\cA$ to $[c]$)}.
\end{definition}

We obtain a useful characterization of retracts.

\begin{lemma} \label{lem:c-slice-and-retract}
  Let $\cA$ be an equality language
  and $f$ be a retraction from $\cA$ to $[c]$.
  Then $f(R) = R \cap [c]^{\ar(R)}$ for all $R \in \cA$.
\end{lemma}
\begin{proof}
  First, observe that $f(R) \subseteq R \cap [c]^{\ar(R)}$: indeed, 
  $f$ is an endomorphism, so $f(R) \subseteq R$, and
  $f(R) \subseteq [c]^{\ar(R)}$ because 
  the range of $f$ is $[c]$.
  Moreover, we have 
  $R \cap [c]^{\ar(R)} \subseteq f(R)$
  because $f$ is constant on $[c]$,
  so it preserves every tuple in $[c]^{\ar(R)}$.
\end{proof}

The finite-domain language $\{ R \cap [c]^{\ar(R)} : R \in \cA \}$ 
is called a \emph{$c$-slice of $\cA$} in~\cite[Section 7]{equalityMinCSParXiv}.
Lemma~\ref{lem:c-slice-and-retract} shows that
a $c$-slice of $\cA$ is the retract $\cA_f$ under 
any retraction $f$ from $\cA$ to $[c]$.
Note that the definition of the $c$-slice does not depend on $f$,
so we can talk about \emph{the retract of $\cA$ to $[c]$}.
We will use this fact implicitly
when transferring results from 
Theorem~57~in~\cite{equalityMinCSParXiv}.

\begin{lemma}[Follows from Theorem~57~in~\cite{equalityMinCSParXiv}]
  \label{lem:retracts}
  Let $\cA$ be an equality language 
  that is $0$-valid and not Horn,
  and let $c$ be a positive integer.
  Then exactly one of the following holds:
  \begin{itemize}
    \item $\cA$ does not retract to $[c]$, and $\csp{\cA^+_c}$ is \NP-hard.
    \item $\cA$ retracts to $[c]$, and $\csp{\cA^+_c}$ is \NP-hard for all $c \geq 2$.
    \item $\cA$ retracts to $[c]$, and both $\csp{\Delta^+_c}$ for the retract $\Delta$
    and $\csp{\cA^+_c}$ are in \Poly.
  \end{itemize}
\end{lemma}

Let $\NEQQ{r} = \{ (t_1,\dots,t_r) \in \Nat^r : |\{t_1,\dots,t_r\}| = r \}$,
i.e. the relation that contains every tuple of arity $r$
with all entries distinct.

\begin{lemma} \label{lem:neq-definition}
$(\star)$
  Let $\cA$ and $\cB$ be equality languages
  and $c \in \Int_+$.
  If $\NEQQ{c+1} \notin \cclone{\cA \cup \cB}_{\leq k}$, then
  every satisfiable instance of $\alien{k}{\cA}{\cB}$ has
  a solution whose range is in $[c]$.
\end{lemma}

\begin{lemma} \label{lem:expansion-to-alien}
$(\star)$
  Let $\cA$, $\cB$ be two equality constraint languages,
  and let $c \in \Int_+$ be an integer.
  $\csp{\cA^+_c}$ is polynomial-time reducible to
  $\alien{k}{\cA}{\cB}$ whenever 
  $\NEQQ{c} \in \cclone{\cA \cup \cB}_{\leq k}$.
\end{lemma}

We are ready to present the classification.

\begin{theorem}
\label{thm:equality-classification}
  Let $\cA$ and $\cB$ be equality languages
  such that $\csp{\cA}$ is in \Poly and
  $\csp{\cA \cup \cB}$ is \NP-hard.
  \begin{enumerate}
    \item \label{case:horn}
    If $\cA$ is Horn, $\alien{}{\cA}{\cB}$ is in \FPT parameterized by
    \numac{}.
    \item \label{case:non-horn} 
    If $\cA$ is not Horn,
    $\alien{}{\cA}{\cB}$ is \paraNP-hard 
    parameterized by \numac{}.
    Moreover, there exists an integer $c = c(\cA)$ such that
    $\alien{k}{\cA}{\cB}$ is in \Poly whenever
    $\NEQQ{c} \notin \cclone{\cA \cup \cB}_{\leq k}$, 
    and is \NP-hard otherwise.
  \end{enumerate}
\end{theorem}
\begin{proof}
  $\csp{\cA}$ is in \Poly so
  $\cA$ is Horn or $0$-valid.
  If $\cA$ is Horn, then Corollary~\ref{cor:Horn-fpt} applies,
  proving part~\ref{case:horn} of the theorem.
  Suppose $\cA$ is $0$-valid and not Horn.
  By
  applying Lemma~\ref{lem:not-horn-expansion-nph} to $\cA$,
  we infer that there is a minimum positive integer $c$ 
  such that $\csp{\cA^+_c}$ is \NP-hard.
  Since $\cA$ is $0$-valid, we have $c \geq 2$.
  Using Lemma~\ref{lem:expansion-to-alien},
  we can reduce $\csp{\cA^+_c}$ to
  $\alien{k}{\cA}{\cB}$ in polynomial time whenever 
  $\NEQQ{c} \in \cclone{\cA \cup \cB}_{\cB \leq k}$,
  proving that the latter problem is \NP-hard.
  Observe that $\cB$ is not $0$-valid 
  because $\csp{\cA \cup \cB}$ is \NP-hard,
  so $\neq \; \in \cclone{\cB}$ and
  $\NEQQ{c} \in \cclone{\cA \cup \cB}_{\cB \leq k}$ for
  some finite $k \leq \binom{c}{2}$.
  This show the \paraNP-hardness result in part~\ref{case:non-horn}.
  
  To complete the proof of part~\ref{case:non-horn},
  it suffices to show that we can solve
  $\alien{k}{\cA}{\cB}$ in polynomial time 
  whenever $\NEQQ{c} \notin \cclone{\cA \cup \cB}_{\cB \leq k}$.
  To this end, observe that, by the choice of $c$, 
  if $c' < c$, then $\csp{\cA^+_{c'}}$ is in \Poly.
  Then, by Lemma~\ref{lem:retracts}, 
  $\cA$ retracts to the finite domain $[c']$,
  and the retract $\Delta$ is such that 
  $\csp{\Delta^+_{c'}}$ is in \Poly.
  We will use the algorithm for $\csp{\Delta^+_{c'}}$ 
  in our algorithm for
  $\alien{k}{\cA}{\cB}$ that works for all $k$ such that 
  $\NEQQ{c} \notin \cclone{\cA \cup \cB}_{\leq k}$
  
  Let $I$ be an instance of $\alien{k}{\cA}{\cB}$.
  Since $\NEQQ{c} \notin \cclone{\cA \cup \cB}_{\cB \leq k}$, 
  Lemma~\ref{lem:neq-definition} implies that
  $I$ is satisfiable if and only if
  it admits a satisfying assignment with range $[c-1]$.
  Let $X$ be the set of variables in $I$ that occur
  in the scopes of the alien constraints.
  Note that $|X| \in O(k)$.
  Enumerate all assignments $\alpha : X \to [c-1]$,
  and check if it satisfies
  all $\cB$-constraints in $I$.
  If not, reject it, otherwise
  remove the $\cB$-constraints and add
  unary constraints $x = \alpha(x)$ for all $x \in X$ instead.
  This leads to an instance of $\csp{\Delta^{+}_{c-1}}$,
  which is solvable in polynomial time.
  If we obtain a satisfiable instance for some $\alpha$,
  then accept $I$, and otherwise reject it.
  Correctness follows by 
  Lemma~\ref{lem:neq-definition}
  and the fact that the algorithm
  considers all assignments from $X$ to $[c]$.
  We make $2^{O(k)}$ calls to the algorithm
  for $\csp{\Delta^+_{c-1}}$, where $k$ is a fixed constant, 
  and each call runs in polynomial time.
  This completes the proof.
\end{proof}

Theorem~\ref{thm:boolean-classification} 
implies that $\alien{}{\cA}{\cB}$ is 
\paraNP-hard if and only if
$\alien{k}{\cA}{\cB}$ is \NP-hard for some $k$, 
and it is in \FPT parameterized
by \numac{} otherwise. 
Theorem~\ref{thm:equality-classification} now implies
a dichotomy for $\textsc{Redundant}(\cdot)$,
$\textsc{Impl}(\cdot)$, and $\textsc{Equiv}(\cdot)$
over finite equality languages.

\begin{theorem}
\label{thm:eq-lang-redundant}
$(\star)$
  Let $\cA$ be a finite equality language.
  Then $\textsc{Redundant}(\cA)$,
  $\textsc{Impl}(\cA)$, and
  $\textsc{Equiv}(\cA)$ are either in \Poly or \NP-hard (under polynomial-time Turing reductions).
\end{theorem}

Algebraically characterizing the exact borderline between tractable and hard 
cases of the problem seems difficult.
In particular, given a $0$-valid non-Horn equality language $\cA$,
answering whether $\alien{1}{\cA}{\bar{\cA}}$ is in \Poly,
i.e. whether $\NEQQ{c} \in \cclone{\cA \cup \bar{R}}_{\leq 1}$ 
for some $R \in \cA$ and large enough $c$, 
requires a deeper understanding of such languages.
However, one can show that the answer to this, and even a more general
question is decidable.

\begin{proposition}
  \label{prop:equality-csp-borderline-decidable}
  ($\star$)
  There is an algorithm that takes 
  two equality constraint languages $\cA$ and $\cB$
  and outputs minimum $k \in \Nat \cup \{\infty\}$ such that
  $\alien{k}{\cA}{\cB}$ is \NP-hard.
\end{proposition}

\section{Discussion} \label{sec:discussion}

We have focused on structures with finite signatures in this paper.
This is common in the CSP literature since relational structures with infinite signature cause vexatious representational issues.
It may, though, be interesting to look at structures with
infinite signatures, too. Zhuk~\cite{Zhuk:icm2022} observes that the complexity of the following problem is open: given a system of linear equations mod 2 and a single linear equation mod 24, find a satisfying assignment over the domain $\{0,1\}$. 
The equations have unbounded arity so this problem can be viewed
as a $\alien{1}{\cA}{\cB}$ problem where $\cA,\cB$ have infinite signatures. 
This question is thus not directly answered by Theorem~\ref{thm:boolean-classification}.
Second, let us also remark that when considering
$\alien{}{\cA}{\cB}$, we have assumed that both $\cA$ and $\cB$
are taken from some nice ``superstructure''. For example, in the
equality language case we assume that both structures are first-order reducts of $(\Nat;=)$. One could choose structures
more freely and, for example, let $\cA$ be an equality language
and $\cB$ a finite-domain language.
This calls for modifications of the underlying theory since (for instance) the algorithm
that Theorem~\ref{thm:fpt-result-variant2} is based on breaks down.

For finite domains we obtained a \emph{coarse} parameterized dichotomy for $\alien{}{\cA}{\cB}$ separating FPT from $\paraNP$-hardness. Sharper results providing the exact borderline between P and NP-hardness for the $\paraNP$-hard cases are required for classifying implication, equivalence, and redundancy. Via Theorem~\ref{thm:core-reduction} and Theorem~\ref{thm:findomcsp} the interesting case is when $\csp{\cA}$ is in P, $\cA \cup \cB$ is core but $\cA$ is not core. This question may be of independent algebraic interest and could be useful for other problems where the core property is not as straightforward as in the CSP case. For example, in \emph{surjective} CSP we require the solution to be surjective, and this problem is generally hardest to analyze when the template is not a core~\cite{DBLP:journals/dam/BodirskyKM12}.

Any complexity classification of the first-order reducts of 
a structure includes by necessity
a classification of equality CSPs. Thus, our equality language classification
lay the foundation for studying first-order reducts of more expressive structures.
A natural step is to study \emph{temporal languages}, i.e.
first-order reducts of $(\Rat;<)$. 
Our classification of
equality constraint languages relies on the work 
in~\cite{bodirsky_chen_pinsker_2010} via~\cite{equalityMinCSParXiv},
who studied the clones of polymorphisms of 
equality constraint languages in more detail.
One important result, due to 
Haddad \& Rosenberg~\cite{haddad1994finite},
is that after excluding several easy cases,
every equality constraint language we end up with
is only closed under operations with range $[c]$ for some constant $c$.
Then, pp-defining the relation $\NEQQ{c+1}$ brings us into \paraNP-hard territory.
Similar characterizations of the polymorphisms for 
reducts of other infinite structures, e.g. $(\Rat; <)$,
would imply corresponding \paraNP-hardness results, and this
appear to be a manageable way forward.

\bibliographystyle{plainurl}%
\bibliography{references}

\newpage

\appendix \label{sec:appendix}

\noindent
\Huge
\textbf{APPENDIX}
\normalsize

\section{Proof of Theorem~\ref{thm:redundant-hybrid}}
\label{app:redundant-hybrid}

\begin{proof} 
Let $I=(V,C)$ be an arbitrary instance of CSP$(\cA)$ with domain $A$.

\medskip

\noindent
\underline{CSP$(\cA)$ is not in \Poly.} We show that \textsc{Redundant}$(\cA)$ is not in \Poly.
Choose
a relation $R\in\cA$ of arity $p>0$ that satisfies $\emptyset \subsetneq R \subsetneq A^p$.
Note that $\cA$ must contain at least one such relation $R$ since otherwise we can trivially
determine whether an instance is a yes-instance or not, and this contradicts that CSP$(\cA)$ is not in \Poly. Let $(t_1,\dots,t_p)$ be an arbitrary tuple in $R$.
We construct another instance $I'=(V',C')$ such that a certain constraint $c\in C'$ is redundant in $I'$ if and only if $I$ is not satisfiable.

\begin{enumerate}
    \item Introduce $p$ fresh variables $y_1,\dots,y_p$ and define $V'=V\cup\{y_1,y_2,...,y_p\}$.
    \item Define the constraint $c=R(y_1,y_2,...,y_p)$ and let $C'=C\cup\{c\}$.
\end{enumerate}
These steps describe a polynomial time reduction from the CSP$(\cA)$ instance $I$ to
the \textsc{Redundant}$(\cA)$ instance $(I',c)$. 
We prove that $I$ is a yes-instance if and only if $(I',c)$ is a no-instance.

If $I$ is satisfiable, then there exists a satisfying assignment $f:V\rightarrow A$ that satisfies all constraints in $C$.
We show that $I'$ is satisfiable by extending the assignment $f$ to $f' : V' \rightarrow A$: let $f'(x)=f(x)$ when $x \in V$ and 
$f'(y_i)=t_i$, $i \in [p]$.
Note that $\Sol((V',C'))\neq \Sol((V',C' \setminus \{c\}))$ since $R\subsetneq D^p$ so
$c$ is not a redundant constraint in $I'$.

If $I$ is not satisfiable, then
$I'$ is not satisfiable since $C \subseteq C'$.
Thus, $\Sol((V',C')) = \Sol((V',C' \setminus \{c\}))$
and
$(I',c)$ is a yes-instance of {\textsc{Redundant}}$(\cA)$.

We conclude that this is a polynomial-time Turing reduction and the lemma follows. Note that
\textsc{Redundant}$(\Gamma)$ is \NP-hard (under polynomial-time Turing reductions) whenever CSP$(\Gamma)$ is \NP-hard.

\medskip

\noindent
\underline{CSP$(\cA)$ is in \Poly.} We show that \textsc{Redundant}$(\cA)$ is in \Poly if and only if
for every relation $R \in \cA$,
$\alien{1}{\cA}{\{\bar{R}\}}$ is in \Poly.

\smallskip

\noindent
\textbf{Right-to-left direction.} Assume $\alien{1}{\cA}{\{\bar{R}\}}$ is in \Poly for every $R \in \cA$. 
For an instance $I=((V,C),c)$ of \textsc{Redundant}$(\cA)$, let $c=R(x_1,\dots,x_k)$ and define $\bar{c}=\bar{R}(x_1,\dots,x_k)$.
Observe that $I' = (V, (C \setminus \{c\}) \cup \bar{c})$ an instance of $\alien{1}{\cA}{\bar{R}}$ and check whether it is satisfiable.
We claim that $I$ is a no-instance if and only $I'$ is satisfiable.
Indeed, $I$ is a no-instance if and only if $\Sol(V, C \setminus \{c\}) \neq \Sol(V, C)$.
Clearly, $\Sol(V, C) \subseteq \Sol(V, C \setminus \{c\})$, so $I$ is a no-instance 
if and only if there is an assignment $\alpha$ that satisfies $C \setminus \{c\}$ and does not satisfy $c$.
Note that such an assignment $\alpha$ satisfies $I' = (V, (C \setminus \{c\}) \cup \bar{c})$,
so it exists if and only if $I'$ is satisfiable.

\noindent
\textbf{Left-to-right direction.} Assume that \textsc{Redundant}$(\cA)$ is in \Poly. 
We show $\alien{1}{\cA}{\bar{R}}$ is in \Poly as well. 
Let $I = (V,C)$ be an instance of the former problem, 
where $c = \bar{R}(x_1,\dots,x_k)$ is in $C$, and let $\bar{c} = R(x_1,\dots,x_k)$. 
Observe that $I' = (V, (C \setminus \{c\}) \cup \bar{c}, \bar{c})$ is an instance of
\textsc{Redundant}$(\cA)$, and check whether it is a yes-instance.
We claim that $I$ is satisfiable if and only if $I'$ is a no-instance.
Indeed, $I'$ is no-instance if and only if
$\Sol(V, (C \setminus \{c\}) \cup \bar{c}) \neq \Sol(V, C \setminus \{c\})$.
Clearly, $\Sol(V, (C \setminus \{c\}) \cup \bar{c}) \subseteq \Sol(V, C \setminus \{c\})$,
so $I'$ is a no-instance if and only if
there exists an assignment $\alpha$ that satisfies $C \setminus \{c\}$
and does not satisfy $\bar{c}$.
Note that such an assignment $\alpha$ satisfies both
$(C \setminus \{c\})$ and $c$, and hence satisfies $I = (V, C)$,
so $\alpha$ exists if and only if $I$ is satisfiable.
\end{proof}

\section{Proof of Theorem~\ref{thm:generalized-ppdefs}}

\begin{proof}
    We only sketch the proof since the details are very similar to the classical reduction for CSPs in Theorem~\ref{thm:pp-red}. The structures $\cA$ and $\cB$ have finite signatures so we can (without loss of generality) assume that we have access to the
    following information: (1)
    the pp-definitions in $\cA$ for the relations in $\cA^* \setminus \cA$, and (2)
    for every $R \in \cB^* \setminus \cB$, 
    a pp-definiton of $R$ in $\cA \cup \cB$ with $k_R$ $\cB$-constraints.
    
    Let $I=(V,C,k)$ be an arbitrary instance
    of $\alien{}{\cA^*}{\cB^*}$.
    We begin by replacing each $(\cA^* \setminus \cA)$-constraint by its precomputed pp-definition in $\cA$. This does not increase the parameter.
     We similarly replace every $(\cB^* \setminus \cB)$-constraint by its pp-definition over $\cA \cup \cB$. There are at most $k$ such constraints in $C$,  and each of them
    is replaced by at most  $k_R$ constraints over $\cB$ for a fixed constant $k_R$.
    This reduction is obviously correct and can be computed in polynomial time.
    The bound on the parameter follows since $k_R$ only
    depends on the chosen pp-definition over the fixed and finite language $\cA \cup \cB$.
    \end{proof}

\section{Proof of Theorem~\ref{thm:core-reduction}}

\begin{proof}
    Let $e$ be an endomorphism with minimal range in $\End(\cA \cup \cB)$, let $\cA' = \{e(R) \mid R \in \cA\}$ and $\cB' = \{e(R) \mid R \in \cB\}$, of the same signature as $\cA$ and $\cB$.
    First, let $(V,C,k)$ be an instance of $\alien{}{\cA}{\cB}$. For each constraint $R(\mathbf{x}) \in C$ we simply replace it by $e(R)(\mathbf{x})$. It is then easy to verify, and well-known, that the resulting instance is satisfiable if and only if $(V,C)$ is satisfiable. Furthermore, observe that if (1) $R \in \cA$ then $e(R) \in \cA'$, and (2) if $R \in \cB$ then $e(R) \in \cB'$. Hence, $(V,C)$ has $k$ alien constraints $R_1(\mathbf{x}_1), \ldots, R_k(\mathbf{x}_k)$ then the new instance has $k$ alien constraints $e(R_1)(\mathbf{x}_1), \ldots, e(R_k)(\mathbf{x}_k)$, too. Hence, it is an fpt-reduction. 
    
    The other direction is similar: let $(V,C_1 \cup C_2,k)$ be an instance of $\alien{}{\cA'}{\cB'}$. For each constraint $e(R)(\mathbf{x}) \in C_1$ we replace it by $R(\mathbf{x})$ for $R \in \cA$, and for each constraint $e(R)(\mathbf{x}) \in C_2$ we replace it by $R(\mathbf{x})$ for $R \in \cB$. Clearly, the number of alien constraints remains unchanged, and the reduction is an fpt-reduction which exactly preserves \numac{}.
\end{proof}

\section{Proof of Lemma~\ref{lem:constant-reduction-stronger}}

\begin{proof}
Let $(V,C)$ be an instance of CSP$(\mathcal{A} \cup \mathcal{C})$. Pick $c \in \cC$ and consider the set of constraints $C^{c} = \{c(x) \mid c \in C\}$. Pick an arbitrary $c(v) \in C^c$ and consider the instance $(V',C')$ obtained by (1) identifying $v'$ with $v$ for any $c(v') \in C^c$ throughout the instance and (2) replacing $C^c$ from the set of constraints with the single constraint $c(v)$. If we repeat this for every $c \in \mathcal{C}$ we obtain an instance of $\alien{|\mathcal{C}|}{\cA}{\cC}$ which is satisfiable if and only if $(V,C)$ is satisfiable.
\end{proof}

\section{Proof of Theorem~\ref{thm:findomcsp}}
\label{sec:proof_of_thm:findomcsp}

\begin{proof}
We use the fact that every structure with finite
domain has a CSP that is either polynomial-time
solvable or \NP-hard~\cite{Bulatov:focs2017,Zhuk:jacm2020}.
Assume that CSP$(\cA \cup \cC_D)$ is in \Poly.
First, we claim that every tuple over $D$ is pp-definable over $\cA \cup \cC_D$. Thus, let $n \geq 1$ and pick $t = (d_1, \ldots, d_n) \in D^n$. It follows that $\{t\}(x_1, \ldots, x_n) \equiv \{d_1\}(x_1) \land \ldots \land \{d_n\}(x_n)$ since each $\{d_i\} \in \cC_D$.
Second, pick an $n$-ary relation $R = \{t_1, \ldots, t_m\} \in \cB$. Since each $\{t_i\} \in \cclone{\cA \cup \cC_D}$, $R$ is a finite union of relations in $\cclone{\cA \cup \cC_D}$, and every relation in $\cB$ is existential positive definable over $\cA \cup \cC_D$.
We conclude that Theorem~\ref{thm:fpt-result-variant2}
is applicable and that 
$\alien{}{\cA}{\cB}$ is in \FPT parameterized by \numac{}.

For the second statement, we assume that CSP$(\cA \cup \cC_D)$ is \NP-hard. 
We show that there is a polynomial-time
reduction from CSP$(\cA \cup \cC_D)$ 
to $\alien{p}{\cA}{\cB}$ for some
$p$ that only depends on $\cB$. First, let $D = \{a_1, \ldots, a_d\}$ and consider the relation 
$E = \{(e(a_1), \ldots, e(a_d)) \mid e \in \End(\cA \cup \cB)\}$, i.e., the set of endomorphisms of $\cA$ viewed as a $d$-ary relation. It is known that $E \in \cclone{\cA \cup \cB}$~\cite[proof of Theorem~17]{DBLP:conf/dagstuhl/BartoKW17} since  $\cA \cup \cB$ is a core. Let $I = (V,C)$ be an instance of $\csp{\cA \cup \mathcal{C}_D}$. By Lemma~\ref{lem:constant-reduction-stronger} we can without loss of generality assume
that $I$ is an instance of $\alien{d}{\cA}{\cC_D}$, and we will produce a polynomial-time reduction to $\alien{1}{\cA}{\{E\}}$ which is sufficient to prove the claim under Theorem~\ref{thm:generalized-ppdefs}.

Let $v_1, \ldots, v_d \in V$  such that $c_i(v_i) \in C$, i.e., the variables being enforced constant values via the constraints in $\mathcal{C}_D$. We remove the constraints $c_1(v_1), \ldots, c_d(v_d)$ and replace them with $E(v_1, \ldots, v_d)$. We claim that the resulting instance $(V, C')$ is satisfiable if and only if $(V,C)$ is satisfiable. First, assume that $f \colon V \to D$ is a satisfying assignment to $(V,C)$. We see that $f(v_i) = c_i$ for each $i \in [d]$ and thus that $(f(v_1) \ldots, f(v_d)) \in E$. For the other direction, assume that $g \colon V \to D$ is a satisfying assignment to $(V,C')$ and consider the function defined by $\pi(a_i) = g(v_i)$ for every $i \in [d]$. Clearly, $(\pi(v_1), \ldots, \pi(v_d)) \in E$, and it follows that $\pi \in \Aut(\cA \cup \cB)$.
Since $\Aut(\cA \cup \cB)$ is an automorphism group it follows that $\pi^{-1} \in \Aut(\cA \cup \cB)$, too, and the function $h(x) = \pi^{-1}(g(x))$ then gives us the required satisfying assignment.
\end{proof}

\section{Proof of Lemma~\ref{lem:one-occurrence-pp}}

\begin{proof}
   By assumption, $c_0 \in \cclone{\mathcal{A}}$, and to simplify the notation we assume that $c_0 \in \mathcal{A}$. This can be done without loss of generality since in the pp-definition below we can replace any occurrence of $c_0$ by its pp-definition. Fix a tuple $(a_1, \ldots, a_n) \in R$ which is not constantly 0. This is possible since $R \neq \emptyset$ and since $R$ is not 0-valid. We then use the definition 
$c_1(x) \equiv \exists y \colon c_0(y) \land R(x_1, \ldots, x_n)$
where $x_i = x$ if $a_i = 1$ and $x_i = y$ if $a_i = 0$.
\end{proof}

\section{Proof of Lemma~\ref{lem:neq-reduction}}

\begin{proof}
Let $(V,C)$ denote an instance of CSP$(\cA \cup \{c_0,c_1\})$.
Assume (without loss of generality by Lemma~\ref{lem:one-occurrence-pp}) that the constant relations $c_0$ and $c_1$
appear at most one time, respectively, in $C$ and that they restrict the variables $z_0$ and $z_1$ as follows: $c_0(z_0)$ and $c_1(z_1)$.
Let $(V,C')$ denote the instance of $\alien{1}{\cA} {\{\neq\}}$ 
where $C'=(C \setminus \{c_0(z_0),c_1(z_1)\}) \cup \{z_0 \neq z_1\}$.
It is not difficult to verify that $(V,C')$ is satisfiable if and only if $(V,C)$ is satisfiable since $\cA$ is invariant under complement.
\end{proof}

\section{Proof of Theorem~\ref{thm:fpt-result-variant}}
\begin{proof}
Condition 3. says that every relation in $\Orb(\cB)$ is a finite union of
relations in $\langle \cA \rangle$ (as pointed out in Section~\ref{sec:general-fpt}).
Condition 2.\ together with the well-known characterization of $\omega$-categorical structures by Engeler, Svenonius, and Ryll-Nardzewski~\cite[Theorem 6.3.1]{Hodges}
imply that every relation in $\cB$ is a finite union of
relations in $\langle \cA \rangle$. We can now apply
Theorem~\ref{thm:fpt-result-variant2}.
\end{proof}

\section{Proof of Lemma~\ref{lem:modelcompletefpt}}

\begin{proof}
The structure $\cB$ is a model-complete core so every relation in
$\Orb(\cA)$ is pp-definable in $\cA$.
Pick an arbitrary relation $R \in \cB$. The structure $\cB$ is
$\omega$-categorical so $R$ is a finite union of relations in $\Orb(\cB)$.
We have assumed that $\Orb(\cB) \subseteq \Orb(\cA)$ so $R$
is existential positive definable in $\cA$.
The result follows from Theorem~\ref{thm:fpt-result-variant}.
\end{proof}

\section{Proof of Lemma~\ref{lem:neq-definition}}
\begin{proof}
  We prove the contrapositive: if there is a satisfiable instance of $\alien{k}{\cA}{\cB}$ 
  with every satisfying assignment taking at least $c$ values,
  then $\cA \cup \cB$ admits a pp-definition of $\NEQQ{c}$
  with $k$ constraints from $\cB$.
  We will use the fact that for every $d$,
  \[ \NEQQ{c}(x_1,\dots,x_c) \equiv \exists x_{c+1},\dots,x_{c+d} \colon \; \NEQQ{d}(x_1,\dots,x_{c+d}), \]
  so it is enough to pp-define a relation $\NEQQ{c'}$ with $c' \geq c$
  to prove the lemma.
  
  Consider a satisfiable instance $I$ of $\alien{k}{\cA}{\cB}$
  as a quantifier-free primitive-positive formula 
  $\phi(x_1,\dots,x_n)$.
  Note that $I$ contains at most $k$ constraints from $\cB$.
  Let $\alpha$ be a satisfying assignment to $I$ with minimum range,
  and assume without loss of generality that the range is $[c]$ for some $c \in \Int_+$.
  We claim that
  $I' = \phi(y_{\alpha(x_1)}, \dots, y_{\alpha(x_n)})$
  is a pp-definition of $\NEQQ{c}$.
  First, note that every injective assignment satisfies $I'$.
  Moreover, every satisfying assignment to $I'$
  also satisfy $I$, so it must take 
  at least $r$ values (i.e. be injective)
  by the choice of $\alpha$.
  Finally, note that $I'$ contains at most $k$ constraints from $\cB$,
  hence it is an instance of $\alien{k}{\cA}{\cB}$.
\end{proof}

\section{Proof of Lemma~\ref{lem:expansion-to-alien}}

\begin{proof}
  Let $I$ be an instance of $\csp{\cA^+}$.
  We construct an equivalent instance $I'$
  of $\alien{k}{\cA}{\cB}$ starting with all constraints
  in $I$ except for the applications of singleton relations,
  i.e. unit assignments.
  Assume without loss of generality that
  $I$ does not contain two contradicting unit assignments.
  To simulate $c$ constants, create 
  variables $x_1, \dots, x_c$ and add 
  the pp-definitions of
  $\NEQQ{c}(x_1,\dots,x_c)$ to $I'$.
  This requires $k$ applications of $\cB$-constraints.
  Now, replace every variable $v$ in $I'$
  such that the constraint $v = i$ is in $I$ 
  with the new variable $x_i$.
  Clearly, the reduction requires polynomial time.
  The correctness follows since we are using a pp-definition
  to simulate relation $\NEQQ{c}$,
  and it can be verified using Theorem~\ref{thm:generalized-ppdefs}.
\end{proof}

\section{Proof of Theorem~\ref{thm:eq-lang-redundant}}

\begin{proof}
  The problems under consideration are 
  equivalent under polynomial-time Turing reductions by Lemma~\ref{lemma:equiveasy}.
  By Theorem~\ref{thm:redundant-hybrid}, 
  $\textsc{Redundant}(\cA)$ is in \Poly if and only if
  $\alien{1}{\cA}{\bar{\cA}}$ is in \Poly,
  where $\bar{\cA} = \{\bar{R} : R \in \cA\}$
  is the language of complements of $\cA$-relations.
  Clearly, if $\cA$ is neither Horn nor $0$-valid,
  then even $\alien{0}{\cA}{\bar{\cA}}$ is \NP-hard,
  implying that $\textsc{Redundant}(\cA)$ is \coNP-hard as pointed out
  after Lemma~\ref{lem:red-hard}.
  If $\cA$ is Horn, then then $\alien{}{\cA}{\cB}$ is in \FPT parameterized by
  \numac{} so $\alien{1}{\cA}{\bar{\cA}}$ is in \Poly,
  and hence $\textsc{Redundant}(\cA)$ is in \Poly.
  If $\cA$ is $0$-valid and not Horn,
  then $\bar{\cA}$ is not $0$-valid and $\csp{\cA \cup \bar{\cA}}$ is \NP-hard.
  Now, Case~\ref{case:non-horn} of Theorem~\ref{thm:equality-classification} applies.
\end{proof}

\section{Proof of Proposition~\ref{prop:equality-csp-borderline-decidable}}
\begin{proof}
  We will assume that the relations are represented by their defining formulas.
  This way, we can use the results of~\cite{bodirsky2013decidability} immediately.
  We can also test inclusion of a tuple in a relation
  compute a representative set of tuples,
  i.e. a set such that every tuple in the relation
  is isomorphic to one member of this set.
  
  We first check whether $\cA$ and $\cB$ are $0$-valid and whether they are Horn.
  For the first, check whether the all-$0$ tuple is in the relation.
  For the second, recall from~\cite[Lemma~8]{Bodirsky:Kara:toct2008} 
  that a relation is Horn if and only if it is closed under
  any binary injective operation.
  Choose an arbitrary binary injective function $f$
  and check that, for every pair of tuples in the representative set,
  the result of applying $f$ to them componentwise is also in the relation.
  To see that this is sufficient, consider an equality relation $R$, 
  two arbitrary tuples $a,b \in R$ and their representatives
  $a', b'$, i.e. tuples in the representative set such that
  $a_i = a_j \iff a'_i = a'_j$ and $b_i = b_j \iff b'_i = b'_j$.
  Then $(a_i, b_i) = (a_j, b_j) \iff (a'_i, b'_i) = (a'_j, b'_j)$,
  so $f(a',b') \in R \implies f(a,b) \in R$.
  If $\cA$ is Horn or both $\cA$ and $\cB$ are $0$-valid,
  then $k = \infty$ by Corollary~\ref{cor:Horn-fpt}.
  Otherwise, $k < \infty$. 
  If $\cA$ is neither Horn nor constant,
  then $\csp{\cA}$ is \NP-hard, and $k = 0$.
  
  The case we are left with is when $\cA$ is constant and not Horn,
  while $\cB$ is not constant.
  By Lemma~\ref{lem:retracts}, there exists $c \in \Nat$
  such that $\csp{\cA^+_c}$ is \NP-hard,
  and $\csp{\cA^+_{c'}}$ is in \Poly for all $c' < c$.
  We show that $c$ can be computed.
  Note that $\csp{\cA^+_1}$ is in \Poly  
  because every instance is satisfiable by a constant assignment.
  Now consider $c = 2$.
  By Theorem~54~in~\cite{osipov2023parameterized}
  and Lemma~\ref{lem:c-slice-and-retract},
  $\csp{\cA^+_2}$ is in \Poly if the $2$-slice of $\cA$
  is preserved by an affine operation,
  and \NP-hard otherwise.
  We can compute the $2$-slice and check whether
  it is closed under an affine operation in polynomial time.
  If $\csp{\cA^+_2}$ is \NP-hard, 
  then $k = 1$ because $\NEQQ{2} \in \cclone{\emptyset \cup \cB}_{\leq k}$.
  Otherwise, proceed to $c \geq 3$.
  Again, using Theorem~54~in~\cite{osipov2023parameterized}
  and Lemma~\ref{lem:c-slice-and-retract},
  we have that $\csp{\cA^+_c}$ for $c \geq 3$ is in \Poly
  if the $c$-slice of $\cA$ is trivial (contains only
  empty or complete relations), and \NP-hard otherwise.
  This can also be checked in polynomial time.

  Now that $c$ is determined,
  by Theorem~\ref{thm:equality-classification},
  $k$ is the minimum integer such that
  $\NEQQ{c} \in \cclone{\cA \cup \cB}_{\leq k}$.
  Note that $k \leq \binom{c}{2}$
  since $\NEQQ{2} \in \cclone{\emptyset \cup \cB}_{\leq k}$.
  We can find minimum $k$ by considering every value 
  $1 \leq t \leq \binom{c}{2}$ in increasing order and 
  checking whether 
  $\NEQQ{c} \in \cclone{\cA \cup \cB}_{\leq t}$.
  Thus, it remains to show that pp-definability of
  $\NEQQ{c}$ in $\cA \cup \cB$ with at most
  $t$ constraints from $\cB$ is decidable.
  To see this, we can view a pp-definition as
  a relation $R \in \cclone{\cA \cup \cB}_{\leq t}$ 
  such that the projection of $R$ onto first $c$
  indices is $\NEQQ{c}$.
  Furthermore, 
  $R(x_1,\dots,x_n) \equiv R_{\cA}(x_1,\dots,x_n) \land R_{\cB}(x_1,\dots,x_n)$,
  where $R_{\cA} \in \cclone{\cA}$ and
  $R_{\cB} \in \cclone{\emptyset \cup \cB}_{\leq t}$.
  Note that $R_{\cB}$ can only depend on $\ell \leq r(\cB) \cdot t$ arguments,
  where $r(\cB)$ is the maximum arity of a relation in $\cB$,
  which is constant.
  The relation $R_{\cA}$ projected onto these $\ell$ arguments
  is an equality relation of arity $\ell$.
  We can guess $\ell$, enumerate all equality relations $R'_{\cA}$
  of arity $\ell$ pp-definable in $\cA$ using~\cite{bodirsky2013decidability} 
  and 
  enumerate all relations $R'_{\cB}$ in $\cclone{\cB}$ definable
  using $t$ constraints,
  and check whether $R'_{\cA}(x_1,\dots,x_c) \land R'_{\cB}(x_1,\dots,x_c)$
  projected onto $x_1,\dots,x_c$ is $\NEQQ{c}$.
  This completes the proof.
\end{proof}

\end{document}